\documentclass[onecolumn,letter]{IEEEtran}

\usepackage{latexsym}
\usepackage{amssymb}
\usepackage{amsmath}
\usepackage{multirow}

\usepackage{color}
\usepackage{bm}
\usepackage[dvips]{graphicx}

\newtheorem{thm}    {Theorem}
\newtheorem{lem}[thm]{Lemma}
\newtheorem{cor}[thm]{Corollary}
\newtheorem{proposition}[thm]{Proposition}
\newtheorem{rem}     {Remark}

 \newenvironment{proofof}[1]{\vspace*{5mm} \par \noindent
         \quad{\it Proof of #1:\hspace{2mm}}}{\endproof
}

\def\sM{\mathsf{M}}

\def\mix{\mathop{\rm mix}}
\def\inv{\mathop{\rm inv}}
\def\normal{\mathop{\rm normal}}

\def\Ker{\mathop{\rm Ker}}

\def\FF{\mathbb F}

\newcommand{\bR}{\mathbb{R}}

\def\cA{{\cal A}}
\def\cM{{\cal M}}
\def\cB{{\cal B}}
\def\cE{{\cal E}}

\def\cM{{\cal M}}

\def\rE{{\rm E}}

\newcommand{\bX}{{\bf X}}

\makeatletter
\newcommand{\lleq}{\mathrel{\mathpalette\gl@align<}}
\newcommand{\ggeq}{\mathrel{\mathpalette\gl@align>}}
\newcommand{\gl@align}[2]{
\vbox{\baselineskip\z@skip\lineskip\z@
\ialign{$\m@th#1\hfil##\hfil$\crcr#2\crcr{}_{{}_{(=)}}\crcr}}}
\makeatother

\def\Label#1{\label{#1}\ [\ \text{#1}\ ]\ }
\def\Label{\label}

\begin{document}
\title{Security analysis of $\varepsilon$-almost dual universal$_2$ hash functions:
smoothing of min entropy vs. smoothing of R\'{e}nyi entropy of order 2}
\author{
Masahito Hayashi
\thanks{
M. Hayashi is with Graduate School of Mathematics, Nagoya University, 
Furocho, Chikusaku, Nagoya, 464-8602, Japan, and
Centre for Quantum Technologies, National University of Singapore, 3 Science Drive 2, Singapore 117542.
(e-mail: masahito@math.nagoya-u.ac.jp)
}}
\date{}
\maketitle

\begin{abstract}
Recently, $\varepsilon$-almost dual universal$_2$ hash functions
has been proposed as a new and wider class of hash functions.
Using this class of hash functions, several efficient hash functions were proposed.
This paper evaluates the security performance 
when we apply this kind of hash functions.
We evaluate the security in several kinds of setting
based on the $L_1$ distinguishability criterion and the modified mutual information criterion.
The obtained evaluation is based on smoothing of 
R\'{e}nyi entropy of order 2 and/or min entropy.
We clarify the difference between these two methods.
\end{abstract}

\section{Introduction}
\subsection{Tight exponential evaluation of $L_1$ distinguishability
under $\varepsilon$-almost dual universality$_2$}
Secure key generation is an important problem in information theoretic security.
When a part of keys are leaked to a third party, we cannot use the key.
In this case, we need to apply a hash function to the keys.
Bennett et al. \cite{BBCM} and H\r{a}stad et al. \cite{ILL} proposed to use universal$_2$ hash functions for privacy amplification
and derived two universal hashing lemma, which provides an upper bound for 
leaked information based on R\'{e}nyi entropy of order $2$.
Two universal hashing lemma can guarantee the security only when the length of the generated keys is less than 
R\'{e}nyi entropy of order $2$.
In order to resolve this drawback,
Renner \cite{Renner} attached the smoothing to min entropy,
which is a lower bound of conditional R\'{e}nyi entropy of order $2$.
The smoothing is the method to replace 
the true distribution by a good distribution that approximates the true distribution.
This method works well when 
the security is evaluated variational distance between the real distribution and the ideal distribution, which is often called the $L_1$ distinguishability criterion.

Now, we consider the case when a random variable $A$ leaked to the third party $E$ is given as $n$-fold independent and identical distribution \cite{AC93,Mau93}.
Under this setting, the optimal asymptotic secure key generation rate is the conditional entropy \cite{AC93,Mau93}. 
The smoothing to min entropy shows that
universal$_2$ hash functions asymptotically achieves the conditional entropy date.
When the key generation rate is smaller than the conditional entropy date,
the $L_1$ distinguishability criterion goes to zero exponentially. 
The previous paper \cite{H-tight} derived an exponentially decreasing rate under 
the universality$_2$.
Its tightness was also shown in \cite{W-H2}.
Note that the importance of exponentially decreasing rate has been explained in the previous papers \cite{H-tight,H-cq}.

Recently, Tsurumaru et al.\cite{Tsuru} proposed to use 
$\varepsilon$-almost dual universal$_2$ hash functions, which is a generalization of liner universal$_2$ hash functions,
and obtained a different version of two universal hashing lemma for this class of hash functions.
Further, the recent paper \cite{H-T} 
proposed several practical hash functions 
under the condition of the $\varepsilon$-almost dual universality$_2$.
The hash functions \cite{H-T}
have a smaller calculation amount and a smaller number of random variables
than the concatenation of Toeplitz matrix and the identity matrix, which is a typical example of universal$_2$ hash functions.
Therefore, 
it is better to evaluate the security 
under the $\varepsilon$-almost dual universality$_2$ rather than under the universality$_2$.
However, the above results in \cite{H-tight,W-H2} were given under the universality$_2$.
In this paper, we show that the above optimal exponential rate can be attained 
by $\varepsilon$-almost dual universal$_2$ hash functions.
Indeed, although the previous paper \cite{H-cq} obtained a similar result in the quantum setting,
the exponent in \cite{H-cq} is strictly worse than the optimal exponent even in the commutative case.

\subsection{Evaluation of modified mutual information}
When the key generation rate is larger than the conditional entropy date,
it is helpful to evaluate how much information is leaked to the third party.
In this case, 
the $L_1$ distinguishability does not go to zero and does not reflect the amount of leaked information properly.
The mutual information seems to work more properly.
Indeed,
many papers \cite{CN,DM,AC93,Mau93,M94,M03,CKbook,Shikata,BTV,YPS,BKS,ZKVB,BB,PB,LLB,WO,NP,NYBNR,TNG,Haya1} 
employ the mutual information as the security criterion.
In the case of secure random number generation,
we need to consider the uniformity as well as the independence.
For this purpose, 
Csisz\'{a}r and Narayan \cite{CN04} modified the mutual information.
Then, we call the criterion the modified mutual information \cite{H-q2,H-cq}.
In the above situation,
the amount of leaked information is expected to increase linearly.
To reflect this requirement, it is natural to surpass the chain rule for the criterion.
In this paper, we show that
only the modified mutual information satisfies 
several natural conditions for our security criteria including the chain rule.
Since these natural conditions for our security criteria 
uniquely determine the security criterion,
only the modified mutual information suits the situation when 
the key generation rate is larger than the conditional entropy date.
Although the previous paper \cite{H-cq} gave a similar characterization in a quantum setting,
the previous characterization \cite{H-cq} could not determine the security criterion uniquely.

When the key generation rate is smaller than the conditional entropy date,
the modified mutual information does not go to zero and increases 
in proportion to the number $n$.
The linear coefficient reflects the amount of leaked information, and is called the equivocation rate.
The previous paper \cite{H-q2} showed that 
the optimal equivocation rate can be attained by universal$_2$ hash functions.
However, it was not shown whether 
the optimal equivocation rate can be attained by $\varepsilon$-almost dual universal$_2$ hash functions.
In this paper, we show that the above optimal equivocation rate 
can be attained by $\varepsilon$-almost dual universal$_2$ hash functions.

Further, due to the Pinsker inequality,
the modified mutual information goes to zero 
when the $L_1$ distinguishability criterion goes to zero.
However, the exponential decreasing rate of the $L_1$ distinguishability criterion 
cannot determine the exponential decreasing rate of 
the modified mutual information because the Pinsker inequality is not so tight.
The previous paper \cite{H-leaked} also derived an lower bound of the exponentially decreasing rate of the modified mutual information
when we apply universal$_2$ hash functions.
In this paper, we show that 
the same lower bound can be attained 
even when we apply $\varepsilon$-almost dual universal$_2$ hash functions.

\subsection{Smoothing of min entropy vs. smoothing of R\'{e}nyi entropy of order $2$}
To discuss the asymptotic performance, 
the paper \cite{Renner} applies the smoothing of the min entropy.
The previous paper \cite{H-tight} applied the smoothing of Renyi entropy of order $2$ when the no leaked information.
Since Renyi entropy of order $2$ gives a better evaluation than the min entropy,
the smoothing of the min entropy cannot surpass
that of Renyi entropy of order $2$. 
The previous paper \cite{H-tight} also showed that the smoothing of the min entropy cannot realize the optimal exponential decreasing rate of the $L_1$ distinguishability criterion 
without any information leakage to the third party.
However, the previous paper \cite{H-tight} did not discuss whether 
the smoothing of the min entropy can realize the optimal exponential decreasing rate of the $L_1$ distinguishability criterion
when a partial information is leaked to the third party.
It is needed to clarify whether 
the smoothing of the min entropy can realize the optimal exponential decreasing rate of the $L_1$ 
distinguishability criterion in this situation
because this general situation is more important from the practical viewpoint
and many people still believe the importance of the smoothing of min entropy.

On the other hand, 
recently, many researchers are interested in second order analysis \cite{strassen,Hsec,H08,Pol,Kont}.
Since the papers \cite{Hsec,H08} for second order analysis employ the method of information spectrum,
which has been established by Han and V\'{e}rdu in their seminal papers \cite{sp2,sp3,sp4,sp5,Han-source} and the book \cite{Han},
many people are interested in how powerful the method of information spectrum is.
As is explained in Section \ref{cqs4-5}, 
the smoothing of the min entropy is essentially the same as the method of information spectrum\footnote{This argument is true even in the classical case. In the quantum case, there are several variants for information spectrum. Hence, we cannot say that the smoothing of the min entropy is essentially the same as the method of information spectrum.
Indeed, the previous paper \cite{TH} discussed this problem only with fidelity distance.}.
Hence, it is important to clarify the limit of the smoothing of the min entropy.

In this paper, we show that the smoothing of the min entropy cannot realize the optimal exponential decreasing rate of the $L_1$ distinguishability criterion 
even when a partial information leaked to the third party.
Then, we arise another question when the smoothing of the min entropy can realize the optimal asymptotic performance.
To answer this question, we show that the smoothing of the min entropy can attain 
the optimal second order key generation rate when the required the $L_1$ distinguishability criterion is fixed
although the same result with the fidelity distance 
was obtained in the previous paper \cite{TH}.
We also show that 
the smoothing of the min entropy can attain 
the optimal equivocation rate.
Here, we should explain that the smoothing of the min entropy is almost same as 
the method of information spectrum, 
which is a powerful and general tool for information theory.
Information spectrum has been established by Han and V\'{e}rdu in their seminal papers \cite{sp2,sp3,sp4,sp5,Han-source} and the book \cite{Han}.
This method can derive asymptotically tight bounds of the optimal performances of various information processings.

These obtained results are summarized as Table \ref{table1}.

\begin{table}[htb]
  \caption{Summary of obtained results.}
\begin{center}
  \begin{tabular}{|l|c|c|c|} 
\hline
 {setting} & 
\!\!\!\!\begin{tabular}{l}
single-shot \\ 
/asymptotic
\end{tabular}%
& $L_1$  & MMI  \\ \hline
 \multirow{3}{*}{exponent (R\'{e}nyi 2)}
& \multirow{2}{*}{single-shot}
& {(\ref{12-5-1-a}) in Theorem \ref{Lem8}}
& {(\ref{12-5-2-a}) in Theorem \ref{Lem8}}  
\\
& 
&{(\ref{3-26-2}) in Theorem \ref{Lem11}}  
&{(\ref{3-26-3}) in Theorem \ref{Lem12}}  
\\ \cline{2-4}
 & asymptotic
&{(\ref{12-18-6-a}) in Theorem \ref{t3-16-1}} 
&{(\ref{12-18-7-a}) in Theorem \ref{t3-16-1}} 
\\ \cline{1-4}
 \multirow{3}{*}{exponent (min)} 
& \multirow{2}{*}{single-shot} 
&{(\ref{12-5-1-2}) in Theorem \ref{t8-27-1}} 
&{(\ref{12-5-2-2}) in Theorem \ref{t8-27-1}} 
\\
 &
& {(\ref{3-19-11}) in Theorem \ref{L3-19-10}} 
& {(\ref{3-19-13}) in Theorem \ref{L3-19-10}} 
\\ \cline{2-4}
 & asymptotic
& {(\ref{3-17-4}) in Theorem \ref{L3-18-3}} 
& {(\ref{3-17-4b}) in Theorem \ref{L3-18-3}} 
\\ \cline{1-4}
 \multirow{3}{*}{second order (min)}
& \multirow{2}{*}{single-shot}
&{(\ref{12-5-1-2}) in Theorem \ref{t8-27-1}} 
& --  
\\
& 
& {(\ref{3-19-11b}) in Theorem \ref{L3-19-10b}} 
&--
\\ \cline{2-4}
 & asymptotic
&{(\ref{12-18-6-a}) in Theorem \ref{t3-16-b}} 
& --
\\ \cline{1-4}
 \multirow{3}{*}{equivocation (min)}
& \multirow{2}{*}{single-shot}
& --  
&{(\ref{12-5-2-d}) in Theorem \ref{t8-27-2}} 
\\
& 
&--
& {(\ref{8-28-1}) in Theorem \ref{t8-27-3}} 
\\ \cline{2-4}
& asymptotic
& --
&{(\ref{8-28-15}) in Theorem \ref{t8-27-4}} 
\\ \hline
  \end{tabular}
\end{center}

\vspace{2ex}

$L_1$ is the $L_1$ distinguishability criterion.
MMI is the modified mutual information criterion. 
(min) means the result 
derived by the smoothing of min entropy.
(R\'{e}nyi 2) means
the results 
derived by the smoothing of R\'{e}nyi entropy of order $2$.
\Label{table1}
\end{table}

\subsection{Significance from information theoretical viewpoint}
Before describing the organization of this paper,
we need to think the current situation of the study of information theoretic security.
Although the information theoretic security has information theoretic formulation,
it has been mainly studied by the community of cryptography not by information theory community.
Further, many important papers \cite{Renner, TSSR11, MDSFT, H-q2,H-cq, TH, MMH, FS08, Tsuru} 
in this direction were written with the quantum terminology.
Since the information theoretic security even with the non-quantum setting
has a sufficient significance from the practical viewpoint
and its formulation has a sufficient similarity to information theory,
it should be studied from information theory more actively.
Indeed, this paper deals with a non-quantum topic.
So, non-quantum researchers should be contained in the reader of this paper.
However, the above mentioned situation obstructs the non-quantum researchers to access the papers in
the information theoretic security even with the non-quantum setting.
To resolve this situation,
this paper needs to contain surveys of results originally obtained in quantum information, which should be written in the non-quantum terminology.

\subsection{Organization}
The remaining part of this paper is organized as follows.
Now, we give the outline of the preliminary parts.
In Section \ref{cqs1}, 
we prepare the information quantities for 
evaluating the security and derive several useful inequalities
for the quantum case.
We also give a clear definition for security criteria.
The contents in Section \ref{cqs1} except for Lemma \ref{L7-1} and Theorem \ref{l8-24-1} are known.
However, since they are given in quantum terminology, these contents are not familiar for people in information theory.
For readers in information theory, their proofs are given in Appendixes.

In Section \ref{cqs3}, 
we introduce several class of hash functions (universal$_2$ hash functions
and $\varepsilon$-almost dual universal$_2$ hash functions).
We clarify the relation between 
$\varepsilon$-almost dual universal$_2$ hash functions
and $\delta$-biased ensemble.
We also derive an $\varepsilon$-almost dual universal$_2$ version of two universal hashing lemma based on 
Lemma for $\delta$-biased ensemble given by Dodis et al \cite{DS05}.
The latter preliminary parts are more technical and used for proofs of the main results.
Although the contents are given the previous paper \cite{Tsuru} with 
terminologies in quantum information,
since they are necessary for the latter discussion,
they are presented in this paper with non-quantum terminologies.

In Section \ref{cqs4}, under 
the $\varepsilon$-almost dual universal$_2$ condition,
we evaluate the $L_1$ distinguishability criterion
and the modified mutual information based on 
the smoothing of min entropy and 
R\'{e}nyi entropy of order $2$.
These parts give the definitions for concepts and quantities describing the main results.
These parts are almost included in the papers \cite{Tsuru,H-cq}.
So, the larger part of Sections \ref{cqs1} \ref{cqs3}, and \ref{cqs4} 
are surveys with non-quantum terminology.

Next, we outline the main results. 
In Section \ref{cqs4-5}, 
using the tail probability of a proper event, 
we evaluate upper bounds given by
the smoothing of min entropy in Section \ref{cqs4}
with the single-shot setting.
This tail probability plays a central role in information spectrum.
The bounds obtained in this section
have smaller complexity for calculation
than those given in Section \ref{cqs4}.
In Section \ref{s4-1}, 
using the information quantities given in Section \ref{cqs1},
we evaluate upper bounds given in Section \ref{cqs4}.
The bounds obtained in this section
have smaller complexity for calculation
than those given in Sections \ref{cqs4-5} and \ref{cqs4}.
In Section \ref{s4-1-b}, 
we derive an exponential decreasing rate for both criteria 
when we simply apply hash functions.
In Section \ref{s4-5}, 
we also discuss the case when the key generation rate is greater than the conditional entropy rate.

\section{Preparation}\Label{cqs1}

\subsection{R\'{e}nyi relative entropy}
In order to discuss the security problem, 
we prepare several information quantities for sub-distributions $P_A$ 
$Q_A$ on a space ${\cal A}$.
That is, these are assumed to satisfy 
the conditions $P_A(a) \ge 0$ and  $\sum_a P_A(a) \le 1$.
R\'{e}nyi introduced 
R\'{e}nyi relative entropy
\begin{align}
D_{1+s}(P_A \|Q_A) & := \frac{1}{s}\log \sum_{a \in \cA} P_A(a)^{1+s} Q_A(a)^{-s} 
\end{align}
as a generalization of relative entropy
\begin{align}
D(P_A \|Q_A) & :=  \sum_{a\in {\cal A}} P_A(a) \log \frac{P_A(a)}{Q_A(a)} 
\end{align}
When we apply a stochastic matrix $\Lambda$ on $\cA$, 
the information processing inequalities
\begin{align}
D(\Lambda(P_A)\|\Lambda(Q_A)) & \le D(P_A\|Q_A) , \quad
D_{1+s}(\Lambda(P_A)\|\Lambda(Q_A)) \le D_{1+s}(P_A\|Q_A)
\Label{8-21-7} 
\end{align}
hold for $s\in (0,1]$.
Since the map $s\mapsto s D_{1+s}(P_A\|Q_A)$ is convex,
we have the following lemma.

\begin{lem}\Label{L22-c1}
$D_{1+s}(P_A\|Q_A)$ is monotonically increasing for 
$s$ in $(-\infty,0) \cup (0,\infty)$.
\end{lem}

When $P_A$ and $Q_A$ are normalized distributions,
we have $s D_{1+s}(P_A\|Q_A)|_{s=0}=0$.
Hence, the concavity of $s\mapsto s D_{1+s}(P_A\|Q_A)$
implies $\lim_{s \to 0}D_{1+s}(P_A\|Q_A)=D(P_A\|Q_A)$.
Then, Lemma \ref{L22-c1} yields the following lemma.
\begin{lem}\Label{L22-c}
When $P_A$ and $Q_A$ are normalized distributions,
\begin{align}
D_{1-s}(P_A \|Q_A) \le
D(P_A \|Q_A) & \le 
D_{1+s}(P_A \|Q_A)
\end{align}
for $s> 0$ .
\end{lem}

\subsection{Conditional R\'{e}nyi entropy}
\subsubsection{Case of joint sub-distribution}
Next, we prepare the conditional R\'{e}nyi entropy
for a joint sub-distribution $P_{A,E}$ on subsets ${\cal A}$ and ${\cal E}$.
In the following discussion,
the sub-distribution $P_A$ and $P_{A,E}$ is not necessarily normalized, 
and is assumed to satisfy 
the condition $\sum_a P_A(a) \le 1$ or $\sum_{a,e} P_{A,E}(a,e)\le 1$.
For the sub-distributions $P_{A}$ and $P_{A,E}$, we define 
the normalized distributions $P_{A,\normal}$ and $P_{A,E,\normal}$
by $P_{A,\normal}(a):=P_{A}(a)/\sum_{a}P_{A}(a)$ and 
$P_{A,E,\normal}(a,e):=P_{A,E}(a,e)/\sum_{a,e}P_{A,E}(a,e)$.
For a sub-distribution $P_{A,E}$, we define 
the marginal sub-distribution $P_{A}$ on ${\cal A}$
by $P_A(a):= \sum_{e \in {\cal E}} P_{A,E}(a,e)$.
Then, we define the conditional sub-distribution 
$P_{A|E}$ on ${\cal A}$
by $P_{A|E}(a|e):= P_{A,E}(a,e)/P_{E,\normal}(e)$. 
The conditional entropy
is given as
\begin{align*}
H(A|E|P_{A,E}) &:= H(A,E|P_{A,E})-H(E|P_{E,\normal}) .
\end{align*}
When we replace $P_{E,\normal}$ 
by another normalized distribution $Q_E$ on ${\cal E}$,
we can generalize the above quantities.
\begin{align}
H(A|E|P_{A,E}\| Q_E) 
:=& \log |{\cal A}| - D(P_{A,E}\|P_{\mix,\cA} \times Q_{E}) \nonumber \\
=& -\sum_{a,e} P_{A,E}(a,e) \log \frac{P_{A,E}(a,e)}{Q_E(e)}  \nonumber \\
=& H(A|E|P_{A,E}) + D(P_E \|Q_E) \nonumber \\
\ge & H(A|E|P_{A,E}) \Label{12-31-2}, 
\end{align}
where
$P_{\mix,\cA}$ is the uniform distribution on the set that the random variable $A$ takes values in.
By using the R\'{e}nyi relative entropy, 
the conditional R\'{e}nyi entropies 
and the conditional min entropy 
are given in the way relative to $Q_E$ as
\begin{align}
H_{1+s}(A|E|P_{A,E}\| Q_E) 
:=& \log |{\cal A}| - D_{1+s}(P_{A,E}\|P_{\mix,\cA} \times Q_{E}) \nonumber \\
=& \frac{-1}{s} \log 
\sum_{a,e} P_{A,E}(a,e)^{1+s} Q_E(e)^{-s} ,\nonumber\\
H_{\min}(A|E|P_{A,E}\| Q_E)
:=& - \log \max_{(a,e):Q_E(e)>0} \frac{P_{A,E}(a,e)}{Q_E(e)}.
\end{align}
Applying Lemma \ref{L22-c1}, we obtain the following lemma.
\begin{lem}\Label{L11-1}
The quantity $H_{1+s}(A|E|P_{A,E}\|Q_E)$
is monotonically decreasing for $s$ in $(-\infty,0) \cup (0,\infty)$.
\end{lem}

Since 
$\sum_{e}P_{E,\normal}(e) \sum_{a} P_{A|E}(a|e) P_{A,E}(a,e)^{s}Q_E(e)^{-s} 
\le \max_{a,e:P_{E}(e)>0} P_{A,E}(a,e)^s Q_E(e)^{-s}$ for $s >0$,
we have
\begin{align}
H_{1+s}(A|E|P_{A,E}\|Q_E) \ge H_{\min}(A|E|P_{A,E}\|Q_E) .
\end{align}
Taking the limit,
we obtain the equality
\begin{align}
\lim_{s\to +\infty}H_{1+s}(A|E|P_{A,E}\|Q_E) = H_{\min}(A|E|P_{A,E}\|Q_E) .
\end{align}

Due to (\ref{8-21-7}),
when we apply 
an operation $\Lambda$ on ${\cal E}$, 
it does not act on the system ${\cal A}$.
Then,
\begin{align}
H(A|E|\Lambda(P_{A,E}) \|\Lambda (Q_{E})) & \ge H(A|E|P_{A,E}\|Q_{E}) \\
H_{1+s}(A|E|\Lambda(P_{A,E})\|\Lambda (Q_{E})) & \ge H_{1+s}(A|E|P_{A,E}\|Q_{E}) \Label{8-15-12-c} .
\end{align}
In particular, 
the inequalities
\begin{align}
H(A|E|\Lambda(P_{A,E})) &\ge H(A|E|P_{A,E}) \nonumber 
\end{align}
hold. Conversely,
when we apply the function $f$ to the random number $a \in \cA$,
we have
\begin{align}
H(f(A)|E|P_{A,E}) \le H(A|E|P_{A,E}).\Label{8-14-1-c}
\end{align}

Now, we introduce two kinds of conditional R\'{e}nyi entropies
by specifying $Q_E$.
The first type is defined by substituting $P_{E,\normal}$ into $Q_E$
as follows
\begin{align*}
H_{1+s}^{\downarrow}(A|E|P_{A,E})
:=& H_{1+s}(A|E|P_{A,E}\|P_{E,\normal}) \\
=&\frac{-1}{s}\log \sum_{e}P_{E,\normal}(e)\sum_{a} P_{A|E}(a|e)^{1+s} \\
H_{\min}^{\downarrow}(A|E|P_{A,E})
:=& H_{\min}(A|E|P_{A,E}\|P_{E,\normal}) \\
=& - \log \max_{(a,e):P_{E,\normal}(e)>0} P_{A|E}(a|e)
\end{align*}
with $s \in \bR \setminus \{0\}$.
Then, as a special case of \eqref{8-15-12-c}, we have
\begin{align}
H_{1+s}^{\downarrow}(A|E|\Lambda(P_{A,E})) &\ge H_{1+s}^{\downarrow}(A|E|P_{A,E}) \Label{8-15-12} 
\end{align}

The second type is defined as
\begin{align}
H_{1+s}^{\uparrow}(A|E|P_{A,E}):=
\max_{Q_E}  H_{1+s}(A|E|P_{A,E}\| Q_E )
\Label{8-26-8-c}
\end{align}
This quantity has another expression as follows.
\begin{lem}\Label{cor1}
A joint sub-distribution $P_{A,E}$ satisfies
the relation
\begin{align}
H_{1+s}^{\uparrow}(A|E|P_{A,E})
=-\frac{1+s}{s}\log \sum_{e} (\sum_{a} P_{A,E}(a,e)^{1+s})^{\frac{1}{1+s}} 
\end{align}
for $s \in [-1,\infty)\setminus \{0\}$.
The maximum in \eqref{8-26-8-c} can be realized when $Q_E(e)= 
(\sum_{a} P_{A,E}(a,e)^{1+s})^{1/(1+s)} /
\sum_e (\sum_{a} P_{A,E}(a,e)^{1+s})^{1/(1+s)}$.
\end{lem}
For reader's convenience, the proof of Lemma \ref{cor1} is given in Appendix \ref{scor1}.
In information theory, we often employ  
Gallager-type \cite{Gal} function \cite{H-tight}:
\begin{align*}
\phi(s|A|E|P_{A,E})
:=&
\log \sum_{e} (\sum_{a} P_{A,E}(a,e)^{1/(1-s)})^{1-s} \\
=&
\log \sum_{e} P_E(e)(\sum_{a} P_{A|E}(a|e)^{1/(1-s)})^{1-s} .
\end{align*}
The quantity $H_{1+s}^{\uparrow}(A|E|P_{A,E})$ can be expressed as
\begin{align*}
H_{1+s}^{\uparrow}(A|E|P_{A,E})
=-\frac{1+s}{s}\phi(\frac{s}{1+s}|A|E|P_{A,E}).
\end{align*}

Although $H_{1+s}^{\uparrow}(A|E|P_{A,E})$
can be lowerly bounded by $H_{1+s}^{\downarrow}(A|E|P_{A,E})$
due to the definition, 
we have the opposite inequality as follows. 

\begin{lem}\Label{cor}
For $s \in [-1,1]\setminus \{0\}$, 
a joint sub-distribution $P_{A,E}$ satisfies
the relation
\begin{align}
H_{1+s}^{\downarrow}(A|E|P_{A,E}) \ge H_{\frac{1}{1-s}}^{\uparrow}(A|E|P_{A,E}).
\Label{8-26-8-k}
\end{align}
The equality holds 
only when 
$P_{A|E=e}$ is uniform distribution for all $e \in {\cal E}$.
\end{lem}
Although Lemma \ref{cor} can be regarded as a special case of (47) or (48) of \cite{MMH}\footnote{Historically, the earlier version of this paper showed Lemma \ref{cor} at the first time.
Then, the paper \cite{MMH} extended this inequality to the quantum setting.},
we give its proof in Appendix \ref{scor} for reader's convenience
because the proof in \cite{MMH} given in quantum terminology.

\subsubsection{Case of joint normalized distribution}
When $P_{A,E}$ is a joint normalized distribution,
the additional useful properties hold as follows.
In this case, 
since 
$\lim_{s\to 0}s H_{1+s}^{\downarrow}(A|E|P_{A,E})=0$,
we have 
\begin{align}
\lim_{s\to 0}H_{1+s}^{\downarrow}(A|E|P_{A,E}) &=H(A|E|P_{A,E})\\
\end{align}
Hence, 
we define 
$H_{1}^{\downarrow}(A|E|P_{A,E})$ and $H_{1}^{\uparrow}(A|E|P_{A,E})$ 
to be $H(A|E|P_{A,E})$.
Further, 
applying Lemma \ref{L22-c}, we obtain the following lemma.
\begin{lem}\Label{L11}
When $P_{A,E}$ and $Q_E$ are normalized distributions,
\begin{align}
H_{1-s}(A|E|P_{A,E}\|Q_E) \ge H(A|E|P_{A,E}\|Q_E) 
\ge H_{1+s}(A|E|P_{A,E}\|Q_E)
\end{align}
for $s>0$.
\end{lem}

Similar properties hold for $H_{1+s}^{\uparrow}(A|E|P_{A,E})$ as follows.
\begin{lem}\Label{L7-1}
\begin{align}
\lim_{s\to 0}
H_{1+s}^{\uparrow}(A|E|P_{A,E})
=H(A|E|P_{A,E}).\Label{1-5-1}
\end{align}
The map $s \mapsto s H_{1+s}^{\uparrow}(A|E|P_{A,E})$ is concave and then
the map $s \mapsto H_{1+s}^{\uparrow}(A|E|P_{A,E})$ 
is monotonically decreasing for $s \in (-1,\infty)$.
In particular, when $P_{A|E=e}$ is not a uniform distribution for an element $e \in {\cal E}$,
the map $s \mapsto s H_{1+s}^{\uparrow}(A|E|P_{A,E})$ is strictly concave and then
the map $s \mapsto H_{1+s}^{\uparrow}(A|E|P_{A,E})$ 
is strictly monotonically decreasing for $s \in (-1,\infty)$.
\end{lem}
Lemma \ref{L7-1} will be shown in Appendix \ref{pL7-1}.

Hence, we define $H_{1}^{\uparrow}(A|E|P_{A,E})$ to be $H(A|E|P_{A,E})$.
Then, the relations (\ref{1-5-1}) and (\ref{8-26-8-c}) 
hold even with $s=0$.

\begin{rem}
Iwamoto and Shikata \cite{MS13} discussed conditional R\'{e}nyi entropies in the different notations.
They denote $H_{1+s}^{\downarrow}(A|E|P_{A,E})$ by $R_{1+s}^{\mathsf{H}}(A|E)$ and $H_{1+s}^{\uparrow}(A|E|P_{A,E})$ by $R_{1+s}^{\mathsf{A}}(A|E)$.
They also compare these with other conditional R\'{e}nyi entropies.
Muller-Lennert et al \cite{MDSFT} denoted $H_{1+s}^{\uparrow}(A|E|P_{A,E})$ by $H_{1+s}^{\downarrow}(P_{A,E}|E)$ in the quantum setting.
Iwamoto and Shikata \cite{MS13} pointed out that these quantities do not satisfy the chain rule.
Instead, Muller-Lennert et al \cite[Proposition 7]{MDSFT} showed the inequality 
$H_{1+s}^{\uparrow}(A|E,E'|P_{A,E,E'}) \ge H_{1+s}^{\uparrow}(A,E'|E|P_{A,E,E'})-\log |{\cal E}'|$ for $s \in (-1,\infty)$.
Also, the paper \cite[Corollary 77]{HM} shows the inequality
$H_{1+s(1-s)}(A|E|P_{A,E,E'}) \ge H_{1+s}^{\downarrow}(A,E|P_{A,E,E'})-\log |{\cal E}|$ for $s \in [0,1)$.
\end{rem}

\subsection{Criteria for secret random numbers}\Label{cqs1-2}
\subsubsection{Case of joint sub-distribution}
Next, we introduce criteria for the amount of the information leaked from the secret random number $A$ to $E$ for joint sub-distribution $P_{A,E}$.
Using the $\ell_1$ norm, we can evaluate the secrecy for the state $P_{A,E}$ as follows:
\begin{align}
d_1(A|E|P_{A,E} ):=\| P_{A,E} - P_A \times P_{E} \|_1.
\end{align}
Taking into account the randomness, 
Renner \cite{Renner} employed the $L_1$ distinguishability criteria for security of the secret random number $A$:
\begin{align}
d_1'(A|E|P_{A,E}):=
\| P_{A,E} - P_{\mix,\cA} \times P_{E} \|_1,
\end{align}
which can be regarded as the difference between the true sub-distribution
$P_{A,E}$ and the ideal sub-distribution $P_{\mix,\cA} \times P_{E}$. 
It is known that the quantity is universally composable \cite{R-K}.

Renner\cite{Renner} defined the conditional $L_2$-distance from uniform of $P_{A,E}$
relative to a distribution $Q_E$ on ${\cal E}$:
\begin{align*}
&{d_{2}}(A|E|P_{A,E}\| Q_E )\\
:=&
\sum_{a,e}
(P_{A,E}(a,e)-P_{\mix,\cA}(a)P_E(e))^2 Q_E(e)^{-1} \\
=&
\sum_{a,e}
P_{A,E}(a,e)^2 Q_E(e)^{-1}
- \frac{1}{|{\cal A}|} \sum_{e} P_E(e)^2 Q_E(e)^{-1} \\
=&
e^{-H_{2}(A|E|P_{A,E}\| Q_E )}
- \frac{1}{|{\cal A}|} e^{D_2(P_{A}\| Q_E )}.
\end{align*}
Using this value
and a normalized distribution $Q_E$, 
we can evaluate $d_1'(A|E|P_{A,E} )$ 
as follows \cite[Lemma 5.2.3]{Renner}:
\begin{align}
d_1'(A|E|P_{A,E} )
\le
\sqrt{|{\cal A}|}
\sqrt{{d_{2}}(A|E|P_{A,E}\| Q_E )}.
\end{align}

\subsubsection{Case of joint normalized distribution}
In the remaining part of this subsection, we assume that
$ P_{A,E}$ is a normalized distribution.
The correlation between $A$ and $E$ can be evaluated by the mutual information
\begin{align}
I(A:E|P_{A,E}) &:= D( P_{A,E} \| P_{A}\times P_{E}) .
\end{align}
By using the uniform distribution $P_{\mix,\cA}$ on ${\cal A}$,
Csisz\'{a}r and Narayan \cite{CN04} 
modified the mutual information to 
\begin{align}
I'(A|E|P_{A,E}) &:= D( P_{A,E} \| P_{\mix,\cA} \times P_{E}) ,
\end{align}
which is called the modified mutual information \cite{H-cq,H-q2} and satisfies
\begin{align}
I'(A|E|P_{A,E}) 
= I(A:E|P_{A,E}) +
D(P_A\|P_{\mix,\cA} )
\end{align}
and
\begin{align}
H(A|E|P_{A,E} ) = -I'(A|E|P_{A,E}) +\log |{\cal A}| .
\end{align}
Indeed, the quantity $I(A:E|P_{A,E})$ 
represents the amount of information leaked by $E$,
and the remaining quantity $D(P_A\|P_{\mix,\cA} )$
describes the difference of the random number $A$ from the uniform random number.
So, if the quantity $I'(A|E|P_{A,E})$ is small,
we can conclude that the random number $A$ has less correlation with $E$ 
and is close to the uniform random number.

Indeed, it is natural to adopt a quantity expressing the difference between the true distribution and the ideal distribution 
$P_{\mix,\cA} \times P_{E}$ as a security criterion.
However, there are several quantities expressing the difference between two distributions.
Both $d_1'(A|E|P)$ and $I'(A|E|P)$
are characterized in this way.
Here, we show that the modified mutual criterion $I'(A|E|P)$ can be derived in a more natural way in the following sense.

It is natural assume the following condition for the security criterion $C(A;E|P)$
as well as the the permutation invariance on ${\cal A}$ and ${\cal E}$. 
\begin{description}
\item[\bf C1]{\bf Chain rule}
$C(A,B|E|P)=C(B|E|P)+C(A|B,E|P)$.

\item[\bf C2]{\bf Linearity}
When the supports of two marginal distributions $P_{E,1}$ and $P_{E,2}$ are disjoint as subsets of ${\cal E}$,
$C(A|E|\lambda P_1+(1-\lambda) P_2)=\lambda C(A|E|P_1)+(1-\lambda)C(A|E|P_2)$.


\item[\bf C3]{\bf Range}
$ \log |\cA| \ge C(A|E|P) \ge 0$.

\item[\bf C4]{\bf Ideal case}
$C(A|E|P_{\mix,\cA} \otimes P_E)=0$.

\item[\bf C5]{\bf Normalization}
$C(A|E||a\rangle \langle a|\otimes P_E)=\log |\cA|$.
\end{description}
Unfortunately, 
the $L_1$ distinguishability does not satisfies {\bf C1} Chain rule.
However, 
we have the following theorem.
\begin{thm}\Label{l8-24-1}
$C(A|E|P)$ satisfies all of the above properties 
if and only if
$C(A|E|P)$ coincides with
the modified mutual information criterion $I'(A|E|P)= \log |\cA| -H(A|E|P)$.
\end{thm}
For a proof, see Appendix \ref{s8-24}. 
Hence, it is natural to adopt
the modified mutual information criterion $I'(A|E|P)$ as a security criterion.
In particular, 
if one emphasizes {\bf C1} Chain rule rather than the universal composability,
it is better to employ the modified mutual information criterion $I'(A|E|P)$.

In particular,
if the quantity $I'(A|E|P_{A,E})$ goes to zero,
$d_1'(A|E|P_{A,E} )$ also goes to zero as follows.
Using Pinsker inequality,
we obtain
\begin{align}
d_1(A|E|P_{A,E} )^2  &\le 2 I(A|E|P_{A,E}) \Label{8-19-14-a} \\
d_1'(A|E|P_{A,E} )^2 &\le 2 I'(A|E|P_{A,E}).\Label{8-19-14}
\end{align}
Conversely, 
we can evaluate $I(A:E|P_{A,E})$ and $I'(A|E|P_{A,E})$ by using $d_1(A|E|P_{A,E} )$ and $d_1'(A|E|P_{A,E} )$
in the following way.
Applying the Fannes inequality, we obtain
\begin{align}
0 \le & I(A:E|P_{A,E}) = H(A| P_{A})+ H(E| P_{E})- H(A,E|P_{A,E}) \nonumber \\
= & H(A,E|P_A \times P_{E} )- H(A,E|P_{A,E}) \nonumber \\
= & \sum_a P_A(a)  H(E|P_{E} )- H(E|P_{E|A=a}) \nonumber \\
\le & \sum_a P_A(a)  \eta ( \| P_{E|A=a} -  P_{E} \|_1,\log |{\cal E}| ) \nonumber \\
= & \eta (  \| P_{E,A} - P_A \times P_{E} \|_1,\log |{\cal E}|) \nonumber \\
=& \eta( d_1(A|E|P_{A,E} ),\log |{\cal E}|),
\Label{8-26-9-a}
\end{align}
where $\eta(x,y):= - x \log x +xy$.
Similarly, we obtain
\begin{align}
& 0 \le I'(A|E|P_{A,E}) \nonumber \\
=& H(A| P_{\mix,\cA} ) + H(E| P_{E} )- H(A,E|P_{A,E}) \nonumber \\
=& H(A,E|P_{\mix,\cA} \times P_{E} )- H(A,E|P_{A,E}) \nonumber\\
=& \sum_e P_E(e) (H(A|P_{\mix,\cA} )- H(A|P_{A|E=e}) ) \nonumber\\
\le & \sum_e P_E(e) ( \| P_{\mix,\cA} - H(A|P_{A|E=e}) \|_1 ,
\log |{\cal A}| ) \nonumber\\
\le & \eta( \| P_{\mix,\cA} \times P_{E}- P_{A,E} \|_1, \log |{\cal A}| )\nonumber \\
=& \eta( d_1'(A|E|P_{A,E} ), \log |{\cal A}| ).
\Label{8-26-9}
\end{align}

\section{Random Hash functions}\Label{cqs3}
\subsection{General random hash functions}
In this section, we focus on a random function $f_{\bX}$ from ${\cal A}$ to ${\cal B}$,
where $\bX$ is a random variable identifying the function $f_{\bX}$.
In this case, the total information of Eve's system is written as $(E,\bX)$.
Then, 
by using 
$P_{f_{\bX}(A),E,\bX}(b,e,x):=\sum_{  a\in f_{\bX}^{-1}(b) }P_{A,E}(a,e)P_{\bX}(x)$,
the $L_1$ distinguishability criterion
is written as
\begin{align}
& d_1'(f_{\bX}(A)|E,\bX|P_{f_{\bX}(A),E,\bX} )\nonumber \\
=&
\|P_{f_{\bX}(A),E,\bX}- P_{\mix,\cB}\times P_{E,\bX}\|_1\nonumber \\
=&
\sum_{x} P_{\bX}(x)
\|P_{f_{\bX=x}(A),E}- P_{\mix,\cB}\times P_{E}\|_1\nonumber \\
=&
\rE_{\bX}
\|P_{f_{\bX}(A),E}- P_{\mix,\cB}\times P_{E}\|_1 .
\end{align}
Also, the modified mutual information is written as
\begin{align}
&I'(f_{\bX}(A)|E,\bX|P_{f_{\bX}(A),E,\bX} )\nonumber \\
=&
D(P_{f_{\bX}(A),E,\bX}\| P_{\mix,\cB}\times P_{E,\bX}) \nonumber\\
=&
\sum_{x} P_{\bX}(x)
D(P_{f_{\bX=x}(A),E,\bX}\|P_{\mix,\cB}\times P_{E}) \nonumber\\
=&
\rE_{\bX}
D(P_{f_{\bX}(A),E,\bX}\| P_{\mix,\cB}\times P_{E}) .
\end{align}

We say that a random function $f_{\bX}$ is {\it $\varepsilon$-almost universal$_2$} \cite{Carter,WC81,Tsuru}, 
if, for any pair of different inputs $a_1$,$a_2$, 
the collision probability of their outputs is upper bounded as
\begin{equation}
{\rm Pr}\left[f_{\bX}(a_1)=f_{\bX}(a_2)\right]
\le \frac{\varepsilon}{|{\cal B}|}.
\Label{eq:def-universal-2}
\end{equation}
The parameter $\varepsilon$ appearing in (\ref{eq:def-universal-2}) is shown to be confined in the region
\begin{equation}
\varepsilon\ge\frac{|\cA|-|{\cal B}|}{|\cA|-1},
\Label{eq:epsilon-lower-bound}
\end{equation}
and in particular, 
a random function $f_{\bX}$ with $\varepsilon=1$ is simply called a {\it universal$_2$} function.

Two important examples of universal$_2$ hash function
are the Toeplitz matrices (see, e.g., \cite{MNP90}), and multiplications over a finite field (see, e.g., \cite{Carter,BBCM}).
A modified form of the Toeplitz matrices is also shown to be universal$_2$, which is given by a concatenation $(X, I)$ of the Toeplitz matrix $X$ and the identity matrix $I$ \cite{H-leaked}.
The (modified) Toeplitz matrices are particularly useful in practice, because there exists an efficient multiplication algorithm using the fast Fourier transform algorithm with complexity $O(n\log n)$ (see, e.g., \cite{MatrixTextbook}).

The following proposition holds for any {\it universal$_2$} function.
\begin{proposition}[{Renner \cite[Lemma 5.4.3]{Renner}}]\Label{Lem5}
Given any joint sub-distribution $P_{A,E}$ on $\cA \times \cE$
and any normalized distribution $Q_E$ on $\cE$,
any universal$_2$ hash function $f_{\bX}$ 
from $\cA$ to $\cM:=\{1, \ldots, \sM\}$
satisfies
\begin{align}
\rE_{\bX} {d_{2}}(f_{\bX}(A)|E|P_{A,E}\|Q_E)
\le
e^{-H_{2}(A|E|P_{A,E}\|Q_E)}.
\end{align}
More precisely, the inequality
\begin{align}
& \rE_{\bX} 
e^{-H_{2}(f_{\bX}(A)|E|P_{A,E} \| Q_E)} 
\nonumber \\
\le &
(1-\frac{1}{\sM} )e^{-H_{2}(A|E|P_{A,E} \| Q_E)} 
 +
\frac{1}{\sM} 
e^{D_2(P_{E} \| Q_E)} 
\end{align}
holds.
\end{proposition}

\subsection{Ensemble of linear hash functions}
Tsurumaru and Hayashi\cite{Tsuru} focus on linear functions over the finite field $\FF_2$. 
Now, we treat the case of linear functions over a finite field $\FF_q$,
where $q$ is a power of a prime number $p$.
We assume that sets ${\cal A}$, ${\cal B}$ are $\FF_q^n$, $\FF_q^m$ respectively with $n\ge m$, 
and $f$ are linear functions over $\FF_q$.
Note that, in this case, there is a kernel $C$ corresponding to a given linear function $f$, 
which is a vector space of 
the dimension $n-m$ or more.
Conversely, when given a vector subspace $C \subset\FF_q^n$ of 
the dimension $n-m$ or more, 
we can always construct a linear function
\begin{equation}
f_{C}: \FF_q^n\to \FF_q^n/C \cong \FF_q^l,\ \ l\le m .
\Label{eq:def-tilde-f}
\end{equation}
That is,
we can always identify a linear hash function $f_{C}$ and a code $C$.

When $C_{\bX}=\Ker f_{\bX}$,
the definition of $\varepsilon$-universal$_2$ function 
(\ref{eq:def-universal-2}) takes the form
\begin{equation}
\forall x\in \FF_q^n\setminus\{0\},\ \ {\rm Pr}\left[f_{\bX}(x)=0\right]\le q^{-m}\varepsilon,
\end{equation}
which is equivalent with
\begin{equation}
\forall x\in \FF_q^n\setminus\{0\},\ \ {\rm Pr}\left[x\in C_{\bX}\right]\le q^{-m}\varepsilon.
\end{equation}
This shows that the kernel $C_{\bX}$ contains sufficient 
information for determining if a random function $f_{\bX}$ 
is $\varepsilon$-almost universal$_2$ or not.

For a given random code $C_{\bX}$,
we define its minimum (respectively, maximum) dimension 
as $t_{\min}:=\min_{\bX}\dim C_{\bX}$ (respectively, $t_{\max}:=\max_{r\in I}\dim C_{\bX}$).
Then, we say that a linear random code $C_{\bX}$ 
of minimum (or maximum) dimension $t$ is an {\it $\varepsilon$-almost universal$_2$} code if the following condition is satisfied
\begin{equation}
\forall x\in \FF_q^n\setminus\{0\},\ \ {\rm Pr}\left[x\in C_{\bX}\right]\le q^{t-n}\varepsilon.
\Label{eq:C-r-upperbound}
\end{equation}
In particular, if $\varepsilon=1$, we call $C_{\bX}$ a {\it universal$_2$} code.


\subsection{Dual universality of a random code}
\Label{sec:approx-duality}
Based on Tsurumaru and Hayashi\cite{Tsuru},
we define several variations of the universality of a error-correcting random code 
and the linear function as follows.
First,
we define the dual random code $C_{\bX}^\perp$ of a given linear random code 
$C_{\bX}$ as the dual code of $C_{\bX}$. 
We also introduce the notion of dual universality as follows.
We say that a random code $C_{\bX}$ in $\FF_q^n$
is {\it $\varepsilon$-almost dual universal$_2$}
with minimum dimension $t$ (with maximum dimension $t$), 
if the dual random code $C_{\bX}^\perp$ is 
$\varepsilon$-almost universal$_2$
with maximum dimension $n-t$ (with minimum dimension $n-t$).
Hence, we say that a linear random function $f_{\bX}$ from $\FF_q^n$ to $\FF_q^m$
is $\varepsilon$-almost dual universal$_2$, 
if the kernels $C_{\bX}$ of $f_{\bX}$ forms an $\varepsilon$-almost dual universal$_2$ 
code with minimum dimension $n-m$.
This condition is equivalent with the condition that
the linear space spanned by the generating matrix of $f_{\bX}$ 
forms an $\varepsilon$-almost universal$_2$ 
random code with maximum dimension $m$.
An explicit example of a dual universal$_2$ function (with $\varepsilon=1$) 
can be given by the modified Toeplitz matrix mentioned earlier \cite{H-qkd2}, 
i.e., a concatenation $(X, I)$ of the Toeplitz matrix $X$ and the identity matrix $I$.
The modified Toeplitz matrix requires $n-1$ bits of random seeds $R$.
This example is particularly useful in practice because it is both universal$_2$ and dual universal$_2$, 
and also because there exists an efficient algorithm with complexity $O(n\log n)$.
When the random variable $R$ is not the uniform random number,
the modified Toeplitz matrix is $q^{n-1}e^{-H_{\min}^{\downarrow}(R)}$-almost dual universal$_2$, as shown in \cite{H-T}.
Therefore, we can evaluate the security of the modified Toeplitz matrix even with non-uniform random seeds.
With these preliminaries, we present the following propositions
in \cite{Tsuru} with non-quantum terminologies and a general prime power $q$:

\begin{proposition}[{\cite[Corollary 2]{Tsuru}}]
\Label{thm:almost-universal2}
An $\varepsilon$-almost universal$_2$ surjective liner random hash function $f_{\bX}$ 
from $\FF_q^n$ to $\FF_q^m$
is $q(1-q^{m}\varepsilon)+(\varepsilon-1)q^{n-m}$-almost dual universal$_2$ liner random hash function.
\end{proposition}
As a special case, we obtain the following.
\begin{cor}
\Label{thm:almost-universal}
Any universal$_2$ linear random function $f_{\bX}$
over a finite filed $\FF_q$ is a $q$-almost dual universal$_2$ function.
\end{cor}

\begin{proposition}[{\cite[Lemma 3]{Tsuru}}]\Label{Lem6-3}
Given a joint sub-distribution $P_{A,E}$ on ${\cal A} \times {\cal E}$
and a normalized distribution $Q_E$ on ${\cal E}$.
When $C_{\bX}$ is an 
$\varepsilon$-almost dual universal$_2$ code
with minimum dimension $t$,
the random hash function $f_{C_{\bX}}$ 
satisfies
\begin{align}
\rE_{\bX} {d_{2}}(f_{C_{\bX}}(A)|E|P_{A,E} \|Q_E )
\le
\varepsilon 
e^{-H_{2}(A|E|P_{A,E} \|Q_E )}.
\Label{12-5-9}
\end{align}
More precisely,
\begin{align}
& \rE_{\bX} 
e^{-H_{2}(f_{C_{\bX}}(A)|E|P_{A,E} \| Q_E)} \nonumber\\
\le &
\varepsilon 
e^{-H_{2}(A|E|P_{A,E} \| Q_E)} 
 +
\frac{1}{q^{n-t}} 
e^{D_2(P_{E} \| Q_E)} 
.\Label{Lem6-3-eq2}
\end{align}
In other words,
an $\varepsilon$-almost dual universal$_2$ function $f_{\bX}$
from $\FF_2^n$ to $\FF_2^{n-t}$ satisfies (\ref{12-5-9}) and (\ref{Lem6-3-eq2}).
\end{proposition}
Since Proposition \ref{Lem6-3} plays an central role instead of Proposition \ref{Lem5}
in this paper
and the proof in the previous paper \cite{Tsuru} is given with 
quantum terminologies and the special case $q=2$,
we give its proof in Appendix \ref{pfLem6-1} without use of quantum terminologies
for reader's convenience.

\section{Security bounds with R\'enyi entropy of order 2 and min entropy}\Label{cqs4}
Firstly, we consider the secure key generation problem from
a common random number $A \in \cA$ which has been partially eavesdropped as an information by Eve.
For this problem, it is assumed that Alice and Bob share a common random number $A \in \cA$,
and Eve has a random number $E$ correlated with the random number $A$, 
whose distribution is $P_E$.
The task is to extract a common random number $f(A)$ from the random number $A \in \cA$, 
which is almost independent of Eve's quantum state.
Here, Alice and Bob are only allowed to apply the same function $f$ to the common random number $A \in \cA$.
Now, we focus on the random function $f_{\bX}$ from 
$\cA$ to $\cM=\{1, \ldots, \sM\}$, where $\bX$ denotes a random variable describing 
the stochastic behavior of the function $f_{\bX}$.

Renner\cite[Lemma 5.2.3]{Renner} essentially
evaluated 
$\rE_{\bX} d_1'(f_{\bX}(A)|E|P_{A,E}) $
by using $\rE_{\bX} {d_{2}}(f_{\bX}(A)|E|P_{A,E}\|Q_E )$ as follows.

\begin{lem}\Label{Lem7-1}
When a state $Q_E$ is a normalized distribution on ${\cal E}$,
any random hash function $f_{\bX}$ 
from $\cA$ to $\{1, \ldots, \sM\}$
satisfies
\begin{align*}
& \rE_{\bX} d_1'(f_{\bX}(A)|E|P_{A,E}) \\
\le & 
\sM^{\frac{1}{2}}
\sqrt{\rE_{\bX} {d_{2}}(f_{\bX}(A)|E|P_{A,E}\|Q_E )} .
\end{align*}
Further, the inequalities used in proof of Renner\cite[Corollary 5.6.1]{Renner} 
imply that
\begin{align*}
&\rE_{\bX} d_1'(f_{\bX}(A)|E|P_{A,E}) \\
\le &
2\|P_{A,E}-P_{A,E}'\|_1
+\rE_{\bX} d_1'(f_{\bX}(A)|E|P_{A,E}') \\
\le &
2\|P_{A,E}-P_{A,E}'\|_1
+\sM^{\frac{1}{2}}
\sqrt{\rE_{\bX} {d_{2}}(f_{\bX}(A)|E|P_{A,E}'\|Q_E)}.
\end{align*}
\end{lem}

Applying the same discussion to Shannon entropy,
we can evaluate the average of
the modified mutual information criterion
by using $\rE_{\bX} {d_{2}}(f_{\bX}(A)|E|P_{A,E}\|Q_E )$ as follows.

\begin{lem}\Label{Lem7}
Assume that 
$P_{A,E}$ is 
a normalized distribution on $\cA \times {\cal E}$.
Any random hash function $f_{\bX}$ 
from $\cA$ to $\cM=\{1, \ldots, \sM\}$
satisfies
\begin{align}
&\rE_{\bX} I'(f_{\bX}(A)|E|P_{A,E}) \nonumber\\
\le &
\log (1+ \sM\rE_{\bX} {d_{2}}(f_{\bX}(A)|E|P_{A,E})
)  \Label{12-6-2}\\
\le &
\sM \rE_{\bX} {d_{2}}(f_{\bX}(A)|E|P_{A,E} \|P_E). \Label{12-6-3}
\end{align}
Further, when a sub-distribution $P_{A,E}'$ satisfies ${P'}_E(e)\le P_E(e)$ for any $e \in {\cal E}$
(we simplify this condition to ${P'}_E \le P_E $),
we obtain
\begin{align}
&\rE_{\bX} I'(f_{\bX}(A)|E|P_{A,E}) \nonumber\\
\le &
\eta(\|P_{A,E}-P_{A,E}'\|_1,\log \sM)
\nonumber\\
& +
\log (1+ \sM \rE_{\bX} {d_{2}}(f_{\bX}(A)|E|P_{A,E}'\|P_E ) )
\Label{12-6-4} \\
\le &
\eta(\|P_{A,E}-P_{A,E}'\|_1,\log \sM)
\nonumber \\
& +
\sM \rE_{\bX} {d_{2}}(f_{\bX}(A)|E|P_{A,E}'\|P_E) ,
\Label{12-6-4-2}
\end{align}
where
$\eta(x,y):= x y -x\log x $.
\end{lem}

\begin{proof}
The inequality $D_2( {P'}_E \| P_E) \le 0$ holds due to the condition ${P'}_E(e)\le P_E(e)$.
Since
\begin{align}
& {d_{2}}(f_{\bX}(A)|E|P_{A,E}' \|P_E) \nonumber \\
= &e^{-H_2(f_{\bX}(A)|E|P_{A,E}' \|P_E )}
-\frac{1}{\sM} e^{D_2( {P'}_E \| P_E)}
\nonumber \\
\ge & e^{-H_2(f_{\bX}(A)|E|P_{A,E}' \|P_E )}
-\frac{1}{\sM} ,
\Label{3-28-2}
\end{align}
we have
\begin{align*}
e^{-H_2(f_{\bX}(A)|E| P_{A,E}' \|P_E)}
\le 
{d_{2}}(f_{\bX}(A)|E|P_{A,E}' \|P_E)
+\frac{1}{\sM}.
\end{align*}
Taking the logarithm, we obtain
\begin{align}
&
-\log \sM  +\log (1+\sM {d_{2}}(f_{\bX}(A)|E|P_{A,E}' \|P_E ) ) \nonumber \\
\ge & -H_2(f_{\bX}(A)|E|P_{A,E}' \|P_E ) 
\ge  - H(f_{\bX}(A)|E|P_{A,E}' \|P_E) . \Label{12-18-3}
\end{align}
Substituting $P_{A,E}$ to $P_{A,E}'$,
we obtain
$H(f_{\bX}(A)|E|P_{A,E}' \|P_E) =H(f_{\bX}(A)|E|P_{A,E})$
and
\begin{align*}
& I'(f_{\bX}(A)|E|P_{A,E}) 
= \log \sM - H(f_{\bX}(A)|E|P_{A,E}) \\
\le &
\log (1+ \sM {d_{2}}(f_{\bX}(A)|E|P_{A,E} ) ).
\end{align*}
Since the function $x \mapsto \log (1+x)$ is concave,
we obtain
\begin{align*}
& \rE_{\bX} I'(f_{\bX}(A)|E|P_{A,E}) \\
\le &
\log (1+
\sM \rE_{\bX} {d_{2}}(f_{\bX}(A)|E|P_{A,E} ) ),
\end{align*}
which implies (\ref{12-6-2}).
The inequality $\log (1+x) \le x$ and (\ref{12-6-2}) yield (\ref{12-6-3}).

Due to Fannes inequality,
the normalized distribution
$P_{A|E=e}(a):=\frac{P_{A,E}(a,e)}{P_E(e)}$
and the sub-distribution
${P'}_{A|E=e}(a):=\frac{P_{A,E}'(a,e)}{P_E(e)}$
satisfy
\begin{align}
& |H(f_{\bX}(A)|P_{A|E=e}) - H(f_{\bX}(A)|{P'}_{A|E=e})| 
\nonumber \\
\le & \eta( \| P_{A|E=e} -{P'}_{A|E=e} \|_1 ,\log \sM).
\Label{12-26-10}
\end{align}
Since
$\sum_e P_E(e) \| P_{A|E=e} -{P'}_{A|E=e} \|_1
= \| P_{A,E}-P_{A,E}' \|_1$,
taking the average under the distribution $P_E$, 
we obtain
\begin{align}
& |H(f_{\bX}(A)|E|P_{A,E}|P_E) - H(f_{\bX}(A)|E|P_{A,E}'|P_E) | \nonumber\\
=& |\sum_e P_E(e) (H(f_{\bX}(A)|P_{A|E=e}) - H(f_{\bX}(A)|{P'}_{A|E=e}) )| \nonumber\\
\le & \sum_e P_E(e)| H(f_{\bX}(A)|P_{A|E=e}) - H(f_{\bX}(A)|{P'}_{A|E=e}) |\nonumber \\
\le & \sum_e P_E(e) \eta( \| P_{A|E=e} -{P'}_{A|E=e} \|_1 ,\log \sM) \nonumber\\
\le & \eta( \sum_e P_E(e) \| P_{A|E=e} -{P'}_{A|E=e} \|_1 ,\log \sM) \nonumber\\
= & \eta( \| P_{A,E}-P_{A,E}' \|_1 ,\log \sM).\Label{12-6-5}
\end{align}
Therefore, using (\ref{12-6-5}) and (\ref{12-18-3}), we obtain
\begin{align*}
& I'(f_{\bX}(A)|E|P_{A,E}) \\
=&
\log \sM - H(f_{\bX}(A)|E|P_{A,E}|P_E) \\
\le & 
\eta( \| P_{A,E}-P_{A,E}' \|_1 ,\log \sM) \\
&+
\log \sM - H(f_{\bX}(A)|E|P_{A,E}'|P_E) \\
\le &
\eta( \| P_{A,E}-P_{A,E}' \|_1 ,\log \sM) \\
&+
\log (1+
\sM {d_{2}}(f_{\bX}(A)|E|P_{A,E}'\|P_E  ) ).
\end{align*}
Taking the expectation of $\bX$ 
and using the concavity of functions $x \mapsto  \eta( x ,\log \sM)$ and $x \mapsto  \log (1+x)$,
we obtain (\ref{12-6-4}).
The inequality $\log (1+x) \le x$ yields (\ref{12-6-4-2}).
In this proof, the condition $P_E(e)' \le P_E(e)$ is crucial because Inequality (\ref{3-28-2}) cannot be shown without this condition.
\end{proof}

Now, we evaluate the security by 
combining Proposition \ref{Lem6-3} and Lemmas \ref{Lem7-1} and \ref{Lem7}.
For this purpose, we introduce the quantities:
\begin{align*}
\Delta_{d,2}(\sM,\varepsilon|P_{A,E})
&:=\min_{Q_E}
\min_{P_{A,E}'}
2 \|P_{A,E}-P_{A,E}'\|_1 
+\sqrt{\varepsilon} \sM^{\frac{1}{2}}
e^{-\frac{1}{2}{H}_{2}(A|E|P_{A,E}'\|Q_E)} \\
&=
\min_{Q_E}
\min_{\epsilon_1>0} 2 \epsilon_1
+\sqrt{\varepsilon} \sM^{\frac{1}{2}}
e^{-\frac{1}{2}{H}_{2}^{\epsilon_1}(A|E|P_{A,E}\|Q_E)}\\
&=
\min_{Q_E}
\min_{R}
2 \min_{P_{A,E}': {H}_{2}(A|E|P_{A,E}'\|Q_E) \ge R} \|P_{A,E}-P_{A,E}'\|_1 
+\sqrt{\varepsilon} \sM^{\frac{1}{2}} e^{-\frac{1}{2}R},\\
\Delta_{I,2}(\sM,\varepsilon|P_{A,E})
&:=
\min_{P_{A,E}': P_{E}'\le P_E}
\eta( \|P_{A,E}-P_{A,E}'\|_1  ,\log \sM)
+ \varepsilon \sM e^{-{H}_{2}(A|E|P_{A,E}' \|P_E )} \\
&=
\min_{\epsilon_1>0} \eta( \epsilon_1 ,\log \sM) 
+ \varepsilon \sM e^{-{H}_{2}^{\downarrow,\epsilon_1}(A|E|P_{A,E})} \\
&=
\min_{R}
\eta( \min_{P_{A,E}': P_{E}'\le P_E, {H}_{2}(A|E|P_{A,E}' \|P_E )\ge R} 
\|P_{A,E}-P_{A,E}'\|_1  ,\log \sM)
+ \varepsilon \sM e^{-R},
\end{align*}
where
\begin{align}
{H}_{2}^{\downarrow,\epsilon_1}(A|E|P_{A,E}\|Q_E)
:=&
\max_{P_{A,E}': \|P_{A,E}-P_{A,E}'\|_1 \le \epsilon_1 } {H}_{2}(A|E|P_{A,E}'\| Q_E) \\
{H}_{2}^{\epsilon_1}(A|E|P_{A,E})
:=&
\max_{P_{A,E}': \|P_{A,E}-P_{A,E}'\|_1 \le \epsilon_1, P_{E}'\le P_E 
} {H}_{2}(A|E|P_{A,E}'\|P_E). 
\end{align}
Note that ${H}_{2}^{\downarrow,\epsilon_1}(A|E|P_{A,E})$ is different from ${H}_{2}^{\epsilon_1}(A|E|P_{A,E}\|P_E)$
because the definition of ${H}_{2}^{\downarrow,\epsilon_1}(A|E|P_{A,E})$
has additional constraints for $P_{A,E}'$.
Then, we can evaluate the averages of both security criteria 
under the $\varepsilon$-almost dual universal$_2$ condition.

\begin{thm}\Label{Lem8}
Assume that 
$Q_E$ is a normalized distribution on ${\cal E}$,
$P_{A,E}$ is a sub-distribution on $\cA \times \cE$,
and a linear random hash function $f_{\bX}$ from $\cA$ to $\cM=\{1, \ldots, \sM\}$
is $\varepsilon$-almost dual universal$_2$.
Then, the random hash function $f_{\bX}$ satisfies
\begin{align}
& \rE_{\bX} d_1'(f_{\bX}(A)|E|P_{A,E} ) \nonumber\\
\le & 
\sqrt{\varepsilon} \sM^{\frac{1}{2}}
e^{-\frac{1}{2}H_{2}(A|E|P_{A,E} \|Q_E)},\nonumber\\
& \rE_{\bX} d_1'(f_{\bX}(A)|E|P_{A,E} )\nonumber \\
\le &
\Delta_{d,2}(\sM,\varepsilon|P_{A,E}).
\Label{12-5-1-a}
\end{align}
When $P_{A,E}$ is a normalized joint distribution, it satisfies
\begin{align}
\rE_{\bX} I'(f_{\bX}(A)|E|P_{A,E} ) 
\le &
\log (1+ \varepsilon \sM e^{-H_{2}^{\downarrow}(A|E|P_{A,E} )}) 
\le 
\varepsilon \sM e^{-H_{2}^{\downarrow}(A|E|P_{A,E} )} \\
\rE_{\bX} I'(f_{\bX}(A)|E|P_{A,E} ) 
\le &
\Delta_{I,2}(\sM,\varepsilon|P_{A,E}).
\Label{12-5-2-a}
\end{align}
\end{thm}

While the same evaluations for the $L_1$ distinguishability criterion
under the universal$_2$ condition
has been shown in Renner\cite[Corollary 5.6.1]{Renner},
those for the modified mutual information criterion 
have not been shown even under the universal$_2$ condition.
All of the above evaluations 
under the $\varepsilon$-almost dual universal$_2$ condition
have not been discussed in Renner. 

Since the function $x \mapsto \eta (x,y)$ is concave,
combing Inequality (\ref{8-26-9}), we obtain the following corollary.
\begin{cor}\Label{c3-29-1}
When a linear random hash function $f_{\bX}$ from $\cA$ to $\cM=\{1, \ldots, M\}$
is $\varepsilon$-almost dual universal$_2$,
any joint sub-distribution $P_{A,E}$ on ${\cal A}$ and ${\cal E}$
satisfies
\begin{align}
\rE_{\bX} I'(f_{\bX}(A)|E|P_{A,E} ) 
\le
\eta( 
\Delta_{d,2}(\sM,\varepsilon|P_{A,E})
, \log |{\cal A}| ).
\Label{3-26-1b}
\end{align}
for $s \in (0,1/2]$.
\end{cor} 

Since the function $x \mapsto \sqrt{x}$ is concave,
combing Inequality (\ref{8-19-14}), we obtain the following corollary.
\begin{cor}\Label{c3-29-2}
When a linear random hash function $f_{\bX}$ from $\cA$ to $\cM=\{1, \ldots, \sM\}$
is $\varepsilon$-almost dual universal$_2$,
any joint normalized distribution $P_{A,E}$ on ${\cal A}\times {\cal E}$ satisfy
\begin{align}
\rE_{\bX} d_1'(f_{\bX}(A)|E|P_{A,E} ) 
\le 
\sqrt{2
\Delta_{I,2}(\sM,\varepsilon|P_{A,E})}
\Label{12-5-6b}
\end{align}
for $s \in (0,1/2]$.
\end{cor}

Further, in the case of the universal$_2$ condition,
Renner\cite[Corollary 5.6.1]{Renner} proposed to 
replace ${H}_2(A|E|P_{A,E}' \|Q_E )$
by the min entropy
$H_{\min}(A|E|P_{A,E}' \|Q_E )$
because ${H}_2(A|E|P_{A,E}' \|Q_E )\ge H_{\min}(A|E|P_{A,E}' \|Q_E )$.
Based on $H_{\min}(A|E|P\|Q_E)$,
Renner\cite{Renner} introduced 
$\epsilon_1$-smooth min entropy as
\begin{align}
H_{\min}^{\epsilon_1}(A|E|P_{A,E}\|Q_E)
:=
\max_{ \|P_{A,E}-P_{A,E}'\|_1 \le \epsilon_1 } H_{\min}(A|E|P_{A,E}'\|Q_E).
\end{align}
For the evaluation of $\rE_{\bX} I'(f_{\bX}(A)|E|P_{A,E} )$, 
adding the condition ${P'}_E \le P_E $,
we define 
\begin{align}
H_{\min}^{\downarrow,\epsilon_1}(A|E|P_{A,E})
:=
\max_{ \|P_{A,E}-P_{A,E}'\|_1 \le \epsilon_1,{P'}_E \le P_E } H_{\min}(A|E|P_{A,E}'\|P_E).
\end{align}
As is shown in Lemma \ref{L8-29-1},
$H_{\min}^{\downarrow,\epsilon_1}(A|E|P_{A,E})$ 
equals 
$H_{\min}^{\epsilon_1}(A|E|P_{A,E}\|P_E)$
while the former has an additional constraint. 
Defining the quantities
\begin{align}
\Delta_{d,\min}(\sM,\varepsilon|P_{A,E})
&:=\min_{Q_E}
\min_{P_{A,E}'}
2 \|P_{A,E}-P_{A,E}'\|_1 
+\sqrt{\varepsilon} \sM^{\frac{1}{2}}
e^{-\frac{1}{2}{H}_{\min}(A|E|P_{A,E}'\|Q_E)} \\
&=
\min_{Q_E}
\min_{\epsilon_1 >0}
2 \epsilon_1
+\sqrt{\varepsilon} \sM^{\frac{1}{2}}
e^{-\frac{1}{2}{H}_{\min}^{\epsilon_1}(A|E|P_{A,E}\|Q_E)} \Label{8-29-1}\\
&=
\min_{Q_E}
\min_{R}
2 \min_{P_{A,E}': {H}_{\min}(A|E|P_{A,E}'\|Q_E) \ge R} \|P_{A,E}-P_{A,E}'\|_1 
+\sqrt{\varepsilon} \sM^{\frac{1}{2}} e^{-\frac{1}{2}R},
\Label{8-29-2}
\\
\Delta_{I,\min}(\sM,\varepsilon|P_{A,E})
&:=\min_{Q_E}
\min_{P_{A,E}': P_{E}'\le Q_E,}
\eta( \|P_{A,E}-P_{A,E}'\|_1  ,\log \sM)
+ \varepsilon \sM e^{-{H}_{\min}(A|E|P_{A,E}' \|P_E )} \\
&=
\min_{\epsilon_1 >0}
\eta( \epsilon_1 ,\log \sM ) 
+\varepsilon \sM e^{-{H}_{\min}^{\downarrow,\epsilon_1}(A|E|P_{A,E})} 
\Label{8-29-3}\\
&=
\min_{R}
\eta( 
\min_{P_{A,E}': P_{E}'\le P_E, {H}_{\min}(A|E|P_{A,E}' \|P_E )\ge R} 
\|P_{A,E}-P_{A,E}'\|_1  
,\log \sM)
+ \varepsilon \sM e^{-R},
\Label{8-29-4}
\end{align}
we obtain the following theorem.
\begin{thm}\Label{t8-27-1}
Assume that 
$Q_E$ is a normalized distribution on ${\cal E}$,
$P_{A,E}$ is a sub-distribution on $\cA \times \cE$,
and a linear random hash function $f_{\bX}$ from $\cA$ to $\cM=\{1, \ldots, \sM\}$
is $\varepsilon$-almost dual universal$_2$.
Then, the random hash function $f_{\bX}$ satisfies
\begin{align}
\rE_{\bX} d_1'(f_{\bX}(A)|E|P_{A,E}) 
\le &
\Delta_{d,\min}(\sM,\varepsilon|P_{A,E})
\Label{12-5-1-2}, \\
\rE_{\bX} I'(f_{\bX}(A)|E|P_{A,E} ) 
\le &
\Delta_{I,\min}(\sM,\varepsilon|P_{A,E})
.\Label{12-5-2-2}
\end{align}
\end{thm}

That is, 
$\Delta_{d,\min}(\sM,\varepsilon|P_{A,E})$ and 
$\Delta_{I,\min}(\sM,\varepsilon|P_{A,E})$
are upper bounds for leaked information
in the respective criteria when the smoothing of min entropy is applied.


\section{Relation with information spectrum}\Label{cqs4-5}
Information spectrum 
can derive asymptotically tight bounds
of the optimal performances of various information processings
by using only 
the asymptotic behavior of 
the tail probability, e.g., $P_{A,E}\{ (a,e)| P_{A|E}(a|e) \ge e^{-R}\}$.
Hence, it can be applied without any assumption for information sources.
While information spectrum originally addresses the asymptotic setting,
we bound the performances in the single-shot setting by using 
the tail probability.
We call these upper and lower bounds
single-shot information spectrum bounds.

In this section, we clarify the relation between
the smoothing of min entropy and single-shot information spectrum bounds.
In stead of the smooth min entropy
${H}_{\min}^{\downarrow,\epsilon_1}(A|E|P_{A,E})$,
we consider the bounds 
$\Delta_{d,\min}(\sM,\varepsilon|P_{A,E})$
and
$\Delta_{I,\min}(\sM,\varepsilon|P_{A,E})$
as functions of
$\min_{P_{A,E}': {H}_{\min}(A|E|P_{A,E}' \|Q_E )\ge R} 
\|P_{A,E}-P_{A,E}'\|_1 $
or
$\min_{P_{A,E}': P_{E}'\le P_E, {H}_{\min}(A|E|P_{A,E}' \|P_E )\ge R} 
\|P_{A,E}-P_{A,E}'\|_1 $.
That is, we employ the formulas (\ref{8-29-2}) and (\ref{8-29-4})
rather than (\ref{8-29-1}) and (\ref{8-29-3}).
Then,
we give their relations 
with the tail probability, e.g., $P_{A,E}\{ (a,e)| P_{A|E}(a|e) \ge e^{-R}\}$
as follows.
\begin{lem}\Label{L8-29-1}
\begin{align}
& \min_{P_{A,E}': {H}_{\min}(A|E|P_{A,E}' \|Q_E )\ge R} 
\|P_{A,E}-P_{A,E}'\|_1 \nonumber \\
=&
\min_{P_{A,E}': {H}_{\min}(A|E|P_{A,E}' \|Q_E )\ge R, P_{A,E}' \le P_{A,E}} 
\|P_{A,E}-P_{A,E}'\|_1 \nonumber \\
=&
P_{A,E}\{ (a,e)| P_{A,E}(a,e) > e^{-R} Q_E(e)\}
-
e^{-R} |\cA| P_{\mix, \cA }\times Q_E \{ (a,e)| P_{A,E}(a,e) > e^{-R} Q_E(e)\}.
\Label{10-16}
\end{align}
and
\begin{align}
& (1-\frac{1}{c}) P_{A,E}\{ (a,e)| P_{A,E}(a,e) > c e^{-R} Q_E(e)\} \nonumber \\
\le &
P_{A,E}\{ (a,e)| P_{A,E}(a,e) > e^{-R} Q_E(e)\}
-
e^{-R} |\cA| P_{\mix, \cA }\times Q_E \{ (a,e)| P_{A,E}(a,e) > e^{-R} Q_E(e)\} \nonumber \\
\le &
P_{A,E}\{ (a,e)| P_{A,E}(a,e) > e^{-R} Q_E(e)\}
\Label{10-16-2}
\end{align}
for $c>1$ and $R$.
\end{lem}
Since the condition $P_{A,E}' \le P_{A,E}$ is more restrictive than
$P_{A}' \le P_{A}$,
we see that 
$H_{\min}^{\downarrow,\epsilon_1}(A|E|P_{A,E})=H_{\min}^{\epsilon_1}(A|E|P_{A,E}\|P_E)$.

\begin{proof}
The optimal sub-distribution $P_{A,E}'$ in the first line of (\ref{10-16}) is
given as
\begin{align}
P_{A,E}'(a,e)=
\left\{\begin{array}{ll}
e^{-R} Q_E(e) & \hbox{ if } P_{A,E}(a,e) > e^{-R} Q_E(e) \\
P_{A,E}(a,e) & \hbox{ if } P_{A,E}(a,e) \le e^{-R} Q_E(e) 
\end{array}
\right.
\end{align}
The sub-distribution is the optimal sub-distribution in the second line
of (\ref{10-16}).
Substituting the above sub-distribution in to the first line, 
we obtain the third line of (\ref{10-16}).

Next, we show (\ref{10-16-2}).
Since
$c P_{A,E}\{ (a,e)| P_{A,E}(a,e) > c e^{-R} Q_E(e)\}
\ge e^{-R} |\cA| P_{\mix, \cA }\times Q_E \{ (a,e)| P_{A,E}(a,e) > c e^{-R} Q_E(e)\}$,
we have
\begin{align}
& 
(1-\frac{1}{c}) P_{A,E}\{ (a,e)| P_{A,E}(a,e) > c e^{-R} Q_E(e)\} \nonumber \\
= &
P_{A,E}\{ (a,e)| P_{A,E}(a,e) > c e^{-R} Q_E(e)\} 
-c P_{A,E}\{ (a,e)| P_{A,E}(a,e) > c e^{-R} Q_E(e)\} \nonumber \\
\le &
P_{A,E}\{ (a,e)| P_{A,E}(a,e) > c e^{-R} Q_E(e)\}
-
e^{-R} |\cA| P_{\mix, \cA }\times Q_E \{ (a,e)| P_{A,E}(a,e) > c e^{-R} Q_E(e)\} \nonumber \\
\le &
P_{A,E}\{ (a,e)| P_{A,E}(a,e) >  e^{-R} Q_E(e)\}
-
e^{-R} |\cA| P_{\mix, \cA }\times Q_E \{ (a,e)| P_{A,E}(a,e) >  e^{-R} Q_E(e)\}  
\Label{10-16-2b}\\
\le &
P_{A,E}\{ (a,e)| P_{A,E}(a,e) > e^{-R} Q_E(e)\}, \nonumber
\end{align}
where the inequality (\ref{10-16-2b}) follows from the fact that
the maximum $
\max_{\Omega} P_{A,E}(\Omega) - 
e^{-R} |\cA| P_{\mix, \cA }\times Q_E (\Omega)$
can be realized by the set $\{ (a,e)| P_{A,E}(a,e) >  e^{-R} Q_E(e)\}$.
\end{proof}

Therefore, 
using the formulas (\ref{8-29-2}) and (\ref{8-29-4}),
we obtain the following theorem.
\begin{thm}\Label{L3-19-10b}
The upper bounds 
$\Delta_{d,\min}(\sM,\varepsilon|P_{A,E})$
and
$\Delta_{I,\min}(\sM,\varepsilon|P_{A,E})$
of leaked information
by the smoothing of min entropy
can be evaluated as follows.
\begin{align}
& 
2(1-\frac{1}{c})
\min_{Q_E}
\min_{R'} 
P_{A,E}
\Bigl\{(a,e)\Bigl|\frac{P_{A,E}(a,e)}{Q_E(e)}> c e^{-R'}  \Bigr\}
+\sqrt{\varepsilon} \sM^{\frac{1}{2}} e^{-\frac{1}{2}R'}
 \Label{3-19-10b}\\
\le & 
\Delta_{d,\min}(\sM,\varepsilon|P_{A,E})
\le 
\min_{Q_E}
\min_{R'} 
2 
P_{A,E}
\Bigl\{(a,e)\Bigl|\frac{P_{A,E}(a,e)}{Q_E(e)}> e^{-R'}  
\Bigr\}
+\sqrt{\varepsilon} \sM^{\frac{1}{2}} e^{-\frac{1}{2}R'},
\Label{3-19-11b}\\
& 
(1-\frac{1}{c})
\min_{R'} 
\eta( P_{A,E}\{(a,e)\in {\cal A} \times {\cal E}|{P_{A|E}(a|e)} \ge c e^{- R'}  \}
 ,\log \sM ) 
+\varepsilon \sM e^{-R'} \Label{3-19-13b}\\
\le &
\Delta_{I,\min}(\sM,\varepsilon|P_{A,E})
\le
\min_{R'} 
\eta( P_{A,E}\{(a,e)\in {\cal A} \times {\cal E}|
{P_{A|E}(a|e)} > e^{-R'}  \}
 ,\log \sM ) 
+\varepsilon \sM e^{-R'} 
\Label{3-19-12b}
\end{align}
for $c>1$.
\end{thm}

Theorem \ref{L3-19-10b} explains that 
the bounds
$\Delta_{d,\min}(\sM,\varepsilon|P_{A,E})$
and
$\Delta_{I,\min}(\sM,\varepsilon|P_{A,E})$
by the smoothing of min entropy 
have almost the same values as
the single-shot information spectrum bounds.
Using this characterization, we evaluate 
the bounds
$\Delta_{d,\min}(\sM,\varepsilon|P_{A,E})$
and
$\Delta_{I,\min}(\sM,\varepsilon|P_{A,E})$
in the latter sections.
However, 
the bounds by the smoothing of 
R\'{e}nyi entropy of order 2
can not be characterized in the same way.
This fact seems to indicate the possibility of
the smoothing of R\'{e}nyi entropy of order 2
beyond the smoothing of min entropy.

\section{Secret key generation: Single-shot case}\Label{s4-1}
In order to obtain useful upper bounds,
we need to calculate or evaluate the quantities
$\Delta_{d,2}(\sM,\varepsilon|P_{A,E})^{1/2}$,
$\Delta_{I,2}(\sM,\varepsilon|P_{A,E})^{1/2}$,
$\Delta_{d,\max}(\sM,\varepsilon|P_{A,E})^{1/2}$, and
$\Delta_{I,\max}(\sM,\varepsilon|P_{A,E})^{1/2}$.
We say that their exact value is the {\it smoothing bound}.
Using the smoothing bound of R\'{e}nyi entropy of order 2,
the paper \cite{H-tight}
derived the following proposition.
\begin{proposition}\Label{Lem10}
The inequality
\begin{align}
\Delta_{d,2}(\sM,1|P_{A,E})
\le
3 
\sM^{s} e^{-s H_{\frac{1}{1-s}}^{\uparrow}(A|E|P_{A,E})}
\Label{3-26-1}
\end{align}
holds for $s \in (0,1/2]$.
\end{proposition}

Using the same smoothing bound, 
we obtain the following evaluation.
\begin{lem}\Label{Lem11}
The inequality
\begin{align}
\Delta_{d,2}(\sM,\varepsilon|P_{A,E})
\le
(2+\sqrt{\varepsilon} ) 
\sM^{s} e^{-s H_{\frac{1}{1-s}}^{\uparrow}(A|E|P_{A,E})}
\Label{3-26-2}
\end{align}
holds for $s \in (0,1/2]$.
\end{lem}

Similar to Theorem \ref{Lem8},
we obtain an upper bound for $\Delta_{I,2}(\sM,\varepsilon|P_{A,E})$.
\begin{thm}\Label{Lem12}
The inequality
\begin{align}
\Delta_{I,2}(\sM,\varepsilon|P_{A,E})
\le
\eta(\sM^s e^{-s H_{1+s}^{\downarrow}(A|E|P_{A,E} )},
\varepsilon+\log \sM )
\Label{3-26-3}
\end{align}
holds for $s \in (0,1]$.
\end{thm}

\begin{proof}
For any integer $\sM$, we choose the subset 
$\Omega_{\sM}:=\{P_{A|E}(a|e) > \sM^{-1}  \}$,
and define the sub-distribution $P_{A,E:\sM}$ by
\begin{align*}
P_{A,E:\sM}(a,e):=
\left\{
\begin{array}{ll}
0 & \hbox{ if } (a,e) \in \Omega_{\sM} \\
P_{A,E}(a,e) & \hbox{ otherwise.}
\end{array}
\right.
\end{align*}
For $0 \le s \le 1$, we can evaluate $e^{-H_2(A|E|P_{A,E:\sM}\|P_E)}$  and 
$d_1(P_{A,E},P_{A,E:\sM})$ as
\begin{align}
& e^{-H_2(A|E|P_{A,E:\sM}\|P_E)} 
= 
\sum_{(a,e)\in \Omega_\sM^c} P_{A,E}(a,e)^2 (P_E(e))^{-1} \nonumber \\
\le &
\sum_{(a,e)\in \Omega_\sM^c} P_{A,E}(a,e)^{1+s} (P_E(e))^{-s} \sM^{-(1-s)}
\nonumber \\
\le &
\sum_{(a,e)} P_{A,E}(a,e)^{1+s} (P_E(e))^{-s} \sM^{-(1-s)}
\nonumber \\
=& e^{-s H_{1+s}^{\downarrow}(A|E|P_{A,E})} \sM^{-(1-s)} ,\Label{5-14-8}\\
& \| P_{A,E}-P_{A,E:\sM}\|_1 \nonumber \\
=& 
P_{A,E}(\Omega_\sM)  =
\sum_{(a,e)\in \Omega_\sM} P_{A,E}(a,e) \nonumber \\
\le &
\sum_{(a,e)\in \Omega_\sM} 
(P_{A,E}(a,e))^{1+s} \sM^s (P_E(e))^{-s} \nonumber \\
\le & 
\sum_{(a,e)} 
(P_{A,E}(a,e))^{1+s} \sM^s (P_E(e))^{-s} \nonumber \\
=& \sM^s  e^{-s H_{1+s}^{\downarrow}(A|E|P_{A,E})} .\Label{5-14-9}
\end{align}
Substituting (\ref{5-14-8}) and (\ref{5-14-9}) into (\ref{12-5-2-a}),
we obtain (\ref{12-5-6b})
because 
\begin{align*}
&\eta(\sM^s e^{-s H_{1+s}^{\downarrow}(A|E|P_{A,E})},\varepsilon+\log \sM )\\
=&
\eta(\sM^s e^{-s H_{1+s}^{\downarrow}(A|E|P_{A,E})} ,\log \sM ) 
+\varepsilon \sM^s e^{-s H_{1+s}^{\downarrow}(A|E|P_{A,E})}.
\end{align*}
\end{proof}

In the above proof, we choose $P_{A,E}'$ to be $P_{A,E:\sM}(a,e)$,
we call the smoothing with this particular choice 
the {\it information-spectrum-smoothing bound}
because this type smoothing  bound is used to derive the entropic information spectrum in \cite{TH}.
Indeed, the paper \cite{H-tight} also employed the information-spectrum-smoothing
 bound to derive Proposition \ref{Lem10}.

Further, 
$\Delta_{d,\min}(\sM,\varepsilon|P_{A,E})$
and $\Delta_{I,\min}(\sM,\varepsilon|P_{A,E})$
can be evaluated as follows.
\begin{thm}\Label{L3-19-10}
The upper bounds 
$\Delta_{d,\min}(\sM,\varepsilon|P_{A,E})$
and
$\Delta_{I,\min}(\sM,\varepsilon|P_{A,E})$
of leaked information
by the smoothing  bound of min entropy
can be evaluated as follows.
\begin{align} 
\Delta_{d,\min}(\sM,\varepsilon|P_{A,E})
& \le 
(2+\sqrt{\varepsilon})\min_{0\le s}
e^{\frac{-s H_{1+s}^{\uparrow}(A|E|P_{A,E}) +sR}{1+2s}}
\Label{3-19-11}\\
\Delta_{I,\min}(\sM,\varepsilon|P_{A,E})
& \le 
 \eta( 
\min_{0\le s}
e^{\frac{-s H_{1+s}^{\downarrow}(A|E|P_{A,E}) +sR}{1+s}}
 ,\varepsilon+\log \sM ) 
\Label{3-19-13}.
\end{align}
\end{thm}

Theorem \ref{L3-19-10} gives upper bounds on
$\Delta_{d,\min}(\sM,\varepsilon|P_{A,E})$ and
$\Delta_{I,\min}(\sM,\varepsilon|P_{A,E})$.
The combination of Theorems \ref{L3-19-10b} and \ref{L3-19-10}
shows the performance of the smoothing  bound of min entropy.
Using these bounds, we can show the tight exponential decreasing rates of 
$\Delta_{d,\min}(\sM,\varepsilon|P_{A,E})$ and
$\Delta_{I,\min}(\sM,\varepsilon|P_{A,E})$.

\begin{proof}
Since 
\begin{align}
&  P_{A,E}
\Bigl\{(a,e)\Bigl|\frac{P_{A,E}(a,e)}{Q_E(e)}> e^{-R'}  
\Bigr\} \nonumber \\
=&
\sum_{(a,e):\frac{P_{A,E}(a,e)}{Q_E(e)}> e^{-R'} }P_{A,E}(a,e) 
\nonumber \\
\le &
\sum_{(a,e):\frac{P_{A,E}(a,e)}{Q_E(e)}> e^{-R'} }
P_{A,E}(a,e) \Bigl(\frac{P_{A,E}(a,e)}{Q_E(e)}e^{R'} \Bigr)^s\nonumber \\
\le &
\sum_{(a,e)}
P_{A,E}(a,e) \Bigl(\frac{P_{A,E}(a,e)}{Q_E(e)}e^{R'} \Bigr)^s\nonumber \\
= &
e^{-s H_{1+s}(A|E|P_{A,E}|Q_E)+s R'}, 
\Label{8-29-3b}
\end{align}
choosing $R'= \frac{\log \sM+2s H_{1+s}(A|E|P_{A,E}|Q_E)}{1+2s}$,
we have
\begin{align*}
& 
2 
P_{A,E}
\Bigl\{(a,e)\Bigl|\frac{P_{A,E}(a,e)}{Q_E(e)}> e^{-R'}  
\Bigr\}
+\sqrt{\varepsilon} \sM^{\frac{1}{2}} e^{-\frac{1}{2}R'} \\
\le &
2 
e^{-s H_{1+s}(A|E|P_{A,E}|Q_E)+s R'}
+\sqrt{\varepsilon} \sM^{\frac{1}{2}} e^{-\frac{1}{2}R'} \\
\le &
(2+\sqrt{\varepsilon} ) 
e^{\frac{-(1+s)s H_{1+s}(A|E|P_{A,E}|Q_E) +sR}{1+2s}}.
\end{align*}
Since the above inequality holds for $s \ge 0$, 
Lemma \ref{cor1} yields that
\begin{align*}
& 
\min_{Q_E}
\min_{R'} 
2 
P_{A,E}
\Bigl\{(a,e)\Bigl|\frac{P_{A,E}(a,e)}{Q_E(e)}> e^{-R'}  
\Bigr\}
+\sqrt{\varepsilon} \sM^{\frac{1}{2}} e^{-\frac{1}{2}R'} \\
\le &
\min_{0\le s}
\min_{Q_E}
(2+\sqrt{\varepsilon} ) 
e^{\frac{-(1+s)s H_{1+s}(A|E|P_{A,E}|Q_E) +sR}{1+2s}}\\
=&
(2+\sqrt{\varepsilon} )
\min_{0 \le s}
e^{\frac{-s H_{1+s}^{\uparrow}(A|E|P_{A,E}) +sR}{1+2s}}
\end{align*}
Hence, 
combining (\ref{3-19-11b}), we obtain (\ref{3-19-11}).

Choosing $R'= \frac{\log \sM+s H_{1+s}^{\downarrow}(A|E|P_{A,E})}{1+s}$,
we have
\begin{align*}
& 
\eta( 
P_{A,E}
\Bigl\{(a,e)\Bigl|{P_{A|E}(a|e)}> e^{-R'}  
\Bigr\}
 ,\log \sM ) 
+ \varepsilon \sM e^{-R'} \\
\le &
\eta( 
e^{-s H_{1+s}^{\downarrow}(A|E|P_{A,E})+s R'}
 ,\log \sM ) 
+ \varepsilon \sM e^{-R'} \\
\le &
\eta( 
e^{\frac{-s H_{1+s}^{\downarrow}(A|E|P_{A,E}) +sR}{1+s}}
 ,\log \sM ) 
+\varepsilon e^{\frac{-s H_{1+s}^{\downarrow}(A|E|P_{A,E}) +sR}{1+s}}\\
=& \eta( 
e^{\frac{-s H_{1+s}^{\downarrow}(A|E|P_{A,E}) +sR}{1+s}}
 ,\varepsilon +\log \sM ) .
\end{align*}
Since the above inequality holds for $s \ge 0$, 
we have
\begin{align*}
& 
\min_{R'} 
\eta( 
P_{A,E}
\Bigl\{(a,e)\Bigl|{P_{A|E}(a|e)}> e^{-R'}  
\Bigr\}
 ,\log \sM ) 
+ \varepsilon \sM e^{-R'} \\
\le &
\min_{0\le s}
 \eta( 
e^{\frac{-s H_{1+s}^{\downarrow}(A|E|P_{A,E}) +sR}{1+s}}
 ,\varepsilon +\log \sM ) \\
= &
 \eta( 
\min_{0\le s}
e^{\frac{-s H_{1+s}^{\downarrow}(A|E|P_{A,E}) +sR}{1+s}}
 ,\varepsilon +\log \sM ) ,
\end{align*}
Hence, 
combining (\ref{3-19-12b}), we obtain (\ref{3-19-13}).
\end{proof}

\begin{rem}
Here, we compare the calculation amount of obtained bounds in 
Sections \ref{cqs4}, \ref{cqs4-5}, and \ref{s4-1}.
In order to calculate the bounds 
$\Delta_{d,2}(\sM,\varepsilon|P_{A,E})$,
$\Delta_{I,2}(\sM,\varepsilon|P_{A,E})$,
$\Delta_{d,\min}(\sM,\varepsilon|P_{A,E})$, and
$\Delta_{I,\min}(\sM,\varepsilon|P_{A,E})$
based on the smoothing,
we need calculate the smooth entropies, which contains   
several optimizations.
Hence, the calculation of these bounds requires
at least double optimization process.
Then, they need higher calculation amounts.
In particular, if the block size becomes larger,
their calculation amounts increase heavily.

The bounds given in Section \ref{cqs4-5}
are calculated from the tail probability.
For example, the tail probability 
$P_{A,E} \{(a,e)|{P_{A|E}(a|e)}> e^{-R'} \} $ 
can be characterized 
as the tail probability with respect to the random variable
$\log P_{A|E}(a|e)$
because 
$P_{A,E} \{(a,e)|{P_{A|E}(a|e)}> e^{-R'} \} 
= P_{A,E} \{(a,e)|\log P_{A|E}(a|e)> {-R'} \} $. 
Hence, in the i.i.d. case, this probability
can be calculated by using statistical packages.
While the calculation amount increases with a rise in the block size,
it is not as large as the above cases because 
statistical packages can be used.

The calculation amounts of the bounds given in Section \ref{s4-1}
are quite small.
In particular, in the i.i.d. case, 
the calculation amounts do not depend on the block size.
These bounds have great advantages with respect to their calculation amounts.
\end{rem}

\section{Secret key generation: Asymptotic case}\Label{s4-1-b}
Next, we consider the case when 
the information source is given by the $n$-fold independent and identical distribution $P_{A,E}^n$ 
of $P_{A,E}$, i.e., $P_{A_n,E_n} =P_{A,E}^n$. 
In this case, Ahlswede and Csisz\'{a}r \cite{AC93} showed that
the optimal generation rate
\begin{align*}
& G(P_{AE}) 
:=
\sup_{\{(f_n,\sM_n)\}}
\left\{\left.
\lim_{n\to\infty} \frac{\log \sM_n}{n}
\right|
d_1'(f_{n}(A_n)|E_n | P_{A,E}^n )
\to 0 
\! \right\}
\end{align*}
equals the conditional entropy $H(A|E)$,
where $f_n$ is a function from $\cA^n$ to $\{1,\ldots, \sM_n\}$.
That is, when the generation rate $R= \lim_{n\to\infty} \frac{\log \sM_n}{n}$
is smaller than $H(A|E)$,
the quantity $d_1'(f_{n}(A_n)|E_n | P_{A,E}^n )$ goes to zero.
In order to treat the speed of this convergence,
we focus on the supremum of  
the {\it exponential rate of decrease (exponent)} for 
$d_1'(f_{n}(A_n)|E_n | P_{A,E}^n )$ and
$I'(f_{n}(A_n)|E_n | P_{A,E}^n )=
I(f_{n}(A_n):E_n | P_{A,E}^n )
+
D(P_{f_{n}(A_n)} \| P_{\mix,f_{n}({\cal A}_n)})$ for a given $R$.

Due to (\ref{8-26-9}),
when $d_1'(f_{C_n}(A_n)|E_n | P_{A,E}^n )$ goes to zero,
$I'(f_{C_n}(A_n)|E_n | P_{A,E}^n )$ goes to zero.
Conversely,
due to (\ref{8-19-14}),
when $I'(f_{C_n}(A_n)|E_n | P_{A,E}^n )$ goes to zero,
$d_1'(f_{C_n}(A_n)|E_n | P_{A,E}^n )$ goes to zero.
So, even if we replace the security criterion by $I'(f_{C_n}(A_n)|E_n | P_{A,E}^n )$,
the optimal generation rate does not change.

Now, we consider 
the case when the length of generated keys behaves as
$n H(A|E|P) + \sqrt{n}R$.
It is known in \cite[Subsection II-D]{W-H} that
\begin{align}
\lim_{n \to \infty} 
\min_{f} d_1'(f(A_n)|E_n | P_{A,E}^n ) 
=
2\int_{-\infty}^{R/\sqrt{V(P)}}
\frac{1}{\sqrt{2 \pi}}
e^{-x^2/2} dx
\Label{12-18-7-d}.
\end{align}
Then,
using Theorem \ref{L3-19-10},
we obtain the following theorem.
\begin{thm}\Label{t3-16-b}
We choose a polynomial $P(n)$.
When a random linear function  $f_{\bX^n}$ from $\cA^n$ to $\{1,\ldots, \lfloor e^{n H(A|E|P) + \sqrt{n}R} \rfloor\}$
is $P(n)$-almost dual universal$_2$,
the relations
\begin{align}
\lim_{n \to \infty} 
\rE_{\bX_n} d_1'(f_{\bX^n}(A_n)|E_n | P_{A,E}^n ) 
= 
\lim_{n \to \infty} 
\min_{f} d_1'(f(A_n)|E_n | P_{A,E}^n ) 
=
2\int_{-\infty}^{R/\sqrt{V(P)}}
\frac{1}{\sqrt{2 \pi}}
e^{-x^2/2} dx
\Label{12-18-7-c}
\end{align}
hold, where
we take the minimum 
under the condition that $f$ is a function from 
$\cA^n$ to $\{1,\ldots, \lfloor e^{n H(A|E|P) + \sqrt{n}R} \rfloor\}$
and
$V(P):= \sum_{a,e}P_{A,E}(a,e)
(\log P_{A|E}(a|e)- H(A|E|P))^2$.
\end{thm}

Lemma \ref{t3-16-b} implies that any $P(n)$-almost dual universal$_2$ hash function
realizes the optimality in the sense of the second order asymptotics
when we employ the $L_1$ distinguishability criterion.
This analysis is obtained from the smoothing bound of min entropy.
That is, this analysis does not require 
the smoothing bound of R\'{e}nyi entropy of order 2.
The second order analysis with the mutual information criterion 
is not so easy.
This topic will be discussed in a future paper.

\begin{proof}
We applying (\ref{3-19-11b}) 
in Theorem \ref{L3-19-10b}
with $R'= n H(A|E|P) + \sqrt{n}R + n^{1/4}$.
Then, the central limit theorem guarantees that
\begin{align*}
& \rE_{\bX_n} 
d_1'(f_{\bX^n}(A_n)|E_n | P_{A,E}^n ) 
\le \Delta_{d,\min}(e^{n H(A|E|P) + \sqrt{n}R + n^{1/4}},P(n)|
P_{A,E}^n) \\
\le &2 P_{A,E}^n \{ (a,e) | P_{A|E}^n(a|e) > e^{-n H(A|E|P) - \sqrt{n}R - n^{1/4}} \}
+\sqrt{P(n)}e^{-n^{1/4}/2} \\
\to &
2 \int_{-\infty}^{R/\sqrt{V(P)}}
\frac{1}{\sqrt{2 \pi}}
e^{-x^2/2} dx.
\end{align*}
Since 
$\min_{f} d_1'(f(A_n)|E_n | P_{A,E}^n ) 
\le d_1'(f_{\bX^n}(A_n)|E_n | P_{A,E}^n ) $,
combining (\ref{12-18-7-d}),
we obtain (\ref{12-18-7-c}).
\end{proof}

Now, we proceed to the exponential decreasing rate 
when we choose the key generation rate $R$ is greater than $H(A|E|P)$.
Since the discussion for the exponential decreasing rate is more complex,
more delicate treatment is required.
First, we should remark that the exponential decreasing rate depends on the choice of the security criterion.
Then, we obtain the following theorem. 
\begin{thm}\Label{t3-16-1}
We choose a polynomial $P(n)$.
When a linear random function $f_{\bX^n}$ from $\cA^n$ to $\{1,\ldots, \lfloor e^{nR} \rfloor\}$
is $P(n)$-almost dual universal$_2$,
the relations
\begin{align}
\liminf_{n \to \infty} \frac{-1}{n}\log \rE_{\bX_n} d_1'(f_{\bX^n}(A_n)|E_n | P_{A,E}^n ) 
\ge &
\liminf_{n \to \infty} \frac{-1}{n}\log \Delta_{d,2}(e^{nR},P(n)|P_{A,E}^n)
\ge e_{d}(P_{A,E}|R)
\Label{12-18-6-a} \\
 \liminf_{n \to \infty} \frac{-1}{n}\log \rE_{\bX_n} I'(f_{\bX^n}(A_n)|E_n | P_{A,E}^n ) 
\ge &
\liminf_{n \to \infty} \frac{-1}{n}\log \Delta_{I,2}(e^{nR},P(n)|P_{A,E}^n)
\ge e_{I}(P_{A,E}|R)
\Label{12-18-7-a}
\end{align}
hold, where
\begin{align}
e_{d}(P_{A,E}|R) &:= 
\max_{0 \le t \le \frac{1}{2}} t (H_{\frac{1}{1-t}}^{\uparrow}(A|E|P_{A,E})-R) \\
e_{I}(P_{A,E}|R) &:=  \max_{0 \le s \le 1} s ( H_{1+s}^{\downarrow}(A|E|P_{A,E} ) -R).
\end{align}
\end{thm}
\begin{proof}
(\ref{12-18-6-a}) can be shown by Theorem \ref{Lem11}.
(\ref{12-18-7-a}) can be shown by Theorem \ref{Lem12}.
\end{proof}

As is shown in Appendix \ref{aL3-29-1}, the following relation between two exponents $e_{I}(P_{A,E}|R)$ and $e_{d}(P_{A,E}|R)$ holds.
\begin{lem}\Label{L3-29-1}
we obtain
\begin{align}
\frac{1}{2}e_{I}(P_{A,E}|R) \le & e_{d}(P_{A,E}|R) \Label{12-21-30} \\
e_{I}(P_{A,E}|R) \ge & e_{d}(P_{A,E}|R) .\Label{12-21-31}
\end{align}
\end{lem}

First, we consider the tightness of Inequality (\ref{12-18-6-a}).
Corollary \ref{c3-29-2} yields the exponent $\frac{e_{I}(P_{A,E}|R)}{2}$
for the $L_1$ distinguishability criterion.
Lemma \ref{L3-29-1} shows that
the exponents by Theorem \ref{Lem11} 
is better than that by Corollary \ref{c3-29-2}.
Further, it is also shown in \cite[Theorem 30]{W-H2} 
that there exists a sequence of universal$_2$ functions $f_{\bX^n}$ from $\cA^n$ to $\{1,\ldots, \lfloor e^{nR} \rfloor\}$
such that
\begin{align}
\limsup_{n \to \infty} \frac{-1}{n}\log \rE_{\bX_n} d_1'(f_{\bX^n}(A_n)|E_n | P_{A,E}^n ) 
\le 
\bar{e}_{d}(P_{A,E}|R) ,
\Label{12-18-6x}
\end{align}
where
\begin{align}
\bar{e}_{d}(P_{A,E}|R) 
:= \max_{0 \le t } t (H_{\frac{1}{1-t}}^{\uparrow}(A|E|P_{A,E})-R). 
\end{align}
When the maximum $\max_{0 \le t } t (H_{\frac{1}{1-t}}^{\uparrow}(A|E|P_{A,E})-R)$ is attained with $t \in (0,\frac{1}{2}]$,
we have 
${e}_{d}(P_{A,E}|R) =\bar{e}_{d}(P_{A,E}|R)$.
Assume that $P(n) \ge 1$.
Then, Since 
$\Delta_{d,2}(e^{nR},1) \le \Delta_{d,2}(e^{nR},P(n)|P_{A,E}^n)
\le \sqrt{P(n)}\Delta_{d,2}(e^{nR},1|P_{A,E}^n)$,
combining (\ref{3-26-1}), (\ref{12-18-6-a}), and (\ref{12-18-6x}) we have
\begin{align}
\lim_{n \to \infty} \frac{-1}{n}\log \Delta_{d,2}(e^{nR},P(n)|P_{A,E}^n)
=\lim_{n \to \infty} \frac{-1}{n}\log \Delta_{d,2}(e^{nR},1|P_{A,E}^n)
= e_{d}(P_{A,E}|R).
\end{align}
That is, our evaluation (\ref{12-18-6-a}) for $\Delta_{d,2}(e^{nR},P(n)|P_{A,E}^n)$ is sufficiently tight in the large deviation sense.

Next, we consider the tightness of Inequality (\ref{12-18-7-a}).
Corollary \ref{c3-29-1} yields the exponent $e_{d}(P_{A,E}|R)$
for the modified mutual information criterion.
Lemma \ref{L3-29-1} shows that
the exponent by Theorem \ref{Lem12}
is better than that by Corollary \ref{c3-29-1}.
Further, the lower bound of the exponent 
$e_{d}(P_{A,E}|R)$ is the same as that given in the previous paper \cite{H-leaked}
under the universal$_2$ condition.
Since the bound given in \cite{H-leaked} is the best lower bound of the exponent,
our evaluation (\ref{12-18-7-a}) for $\Delta_{I,2}(e^{nR},P(n)|P_{A,E}^n)$ is as good as the existing evaluation \cite{H-leaked} in the large deviation sense.

From the above discussion, we find that
the exponents directly obtained by the smoothing bound of R\'{e}nyi entropy of order 2 are
better than 
the exponents derived from the combination of 
Inequality (\ref{8-19-14})/(\ref{8-26-9}) and
the exponent of the other criterion. 
This fact indicates that we need to choose the smoothing bound
dependently of the security criterion.

\begin{rem}
Now, we consider the relation with the recent paper \cite{TSSR11}
discussing the quantum case as including the non-quantum case.
When ${\cal A}=\FF_q$,
we focus on a $1+P(n) q^{-n+ \lfloor nR \rfloor}$-almost universal$_2$ surjective linear function $f_{\bX^n}$ over the field $\FF_q$ 
from $\FF_q^n$ to $\FF_q^{\lfloor nR \rfloor}$.
Thanks to Proposition \ref{thm:almost-universal2},
the surjective linear random function $f_{\bX^n}$ over the field $\FF_q$ is
$q+P(n)$-almost dual universal$_2$. 
Hence, we obtain (\ref{12-18-6-a}),
which can recover a part of the result by \cite{TSSR11} 
with the case of linear functions in the non-quantum case. 
The paper \cite{TSSR11} showed the security with 
an $\epsilon_n$-almost universal$_2$ hash function
when $\epsilon_n$ approaches to $1$.
Since we assume the surjectivity,
our method cannot recover the result by \cite{TSSR11} with the linear hash function perfectly.
\end{rem}

Now, we clarify how better 
our smoothing  bound of R\'{e}nyi entropy of order 2 is than 
the smoothing  bound of min entropy.
As is shown in Appendix \ref{aL3-18-3}, we obtain the following theorem.
\begin{thm}\Label{L3-18-3}
The relations
\begin{align}
&\lim_{n \to \infty}\frac{-1}{n}\log
\Delta_{d,\min}(e^{nR},\varepsilon|P_{A,E}^n)
 \nonumber \\
=&
\tilde{e}_{d}(P_{A,E}|R)
:=
\max_{0\le s}\frac{s (H_{1+s}^{\uparrow} (A|E|P_{A,E}) -R)}{1+2s}
\Label{3-17-4} \\
&\lim_{n \to \infty}\frac{-1}{n}\log
\Delta_{I,\min}(e^{nR},\varepsilon|P_{A,E}^n)
\nonumber \\
= &
\tilde{e}_{I}(P_{A,E}|R)
:=\max_{0\le s}\frac{s H_{1+s}^{\downarrow}(A|E|P_{A,E}) -sR}{1+s}
\Label{3-17-4b} 
\end{align}
hold.
\end{thm} 

For the comparison of the exponents by the smoothing  bound of min entropy and R\'{e}nyi entropy of order 2,
as is shown in Appendix \ref{aL3-19-11}, we have the following lemma
by using Theorem \ref{L3-19-10}.
\begin{lem}\Label{L3-19-11}
The inequalities
\begin{align}
e_{d}(P_{A,E}|R)
>&
\tilde{e}_{d}(P_{A,E}|R)
\Label{3-18-1}\\
e_{I}(P_{A,E}|R)
>&
\tilde{e}_{I}(P_{A,E}|R)
\Label{3-18-1b}
\end{align}
hold
when 
$P_{A|E=e}$ is not a uniform distribution for an element $e \in {\cal E}$.
The equalities 
$e_{d}(P_{A,E}|R)=\tilde{e}_{d}(P_{A,E}|R)$
and $e_{I}(P_{A,E}|R)=\tilde{e}_{I}(P_{A,E}|R)$ hold
when $P_{A|E=e}$ is a uniform distribution for any element $e \in {\cal E}$.
\end{lem} 

Theorem \ref{L3-18-3} and Lemma \ref{L3-19-11} show that the smoothing bound of min entropy 
cannot attain the exponents $e_{d}(P_{A,E}|R)$ and $e_{I}(P_{A,E}|R)$.
That is, 
the bounds 
$\Delta_{d,2}(e^{nR},\varepsilon|P_{A,E}^n)$
and
$\Delta_{I,2}(e^{nR},\varepsilon|P_{A,E}^n)$
by the smoothing bound of R\'{e}nyi entropy of order 2 are strictly better 
than 
the bounds 
$\Delta_{d,\min}(e^{nR},\varepsilon|P_{A,E}^n)$
and
$\Delta_{I,\min}(e^{nR},\varepsilon|P_{A,E}^n)$
by the smoothing bound of min entropy
in the sense of large deviation.
This fact indicates the importance of 
smoothing bound of R\'{e}nyi entropy of order 2.


In summary, while 
the smoothing bound of min entropy yields the tight bound
in the sense of the second order asymptotics,
the smoothing bound of min entropy cannot yield the tight bound
in the sense of the exponential decreasing rate.

\begin{rem}\Label{rem-com}
Here, we give the relation with the results in the quantum case \cite{H-cq}.
The paper \cite{H-cq} showed that
\begin{align}
\liminf_{n \to \infty} \frac{-1}{n}\log 
\Delta_{d,2}(e^{nR},P(n)|P_{A,E}^n)
\ge &
\max_{0 \le t \le \frac{1}{2}} \frac{t}{2(1-t)} 
(H_{\frac{1}{1-t}}^{\uparrow}(A|E|P_{A,E})-R) 
\Label{12-18-6-c} \\
\liminf_{n \to \infty} \frac{-1}{n}\log 
\Delta_{I,2}(e^{nR},P(n)|P_{A,E}^n)
\ge &
\max_{0 \le s \le 1} \frac{s}{2-s} ( H_{1+s}^{\downarrow}(A|E|P_{A,E} ) -R).
\Label{12-18-7-t}
\end{align}
The RHSs of (\ref{12-18-6-c}) and (\ref{12-18-7-t})
are smaller than 
$e_{d}(P_{A,E}|R)$ and $e_{I}(P_{A,E}|R)$, respectively.
Hence, our result is better in the non-quantum case.
\end{rem}

\section{Equivocation rate of secret key generation}\Label{s4-5}
When the key generation rate $R$ is larger than the conditional entropy $H(A|E|P_{A,E})$,
the leaked information does not go to zero.
In this case, it is natural to consider the rate of the conditional
entropy rate of generated keys or
the rate of the modified mutual information \cite{Wyner}.
The former rate is called the equivocation rate,
and is known to be less than the conditional entropy $H(A|E|P_{A,E})$ \cite{Wyner}.
That is, 
the rate of the modified mutual information is larger than $R-H(A|E|P_{A,E})$.
Now, we show that 
the minimum rate of the modified mutual information $R-H(A|E|P_{A,E})$
can be achieved by an $\varepsilon$-almost dual universal$_2$ hash function.
For this purpose, we employ (\ref{12-6-4}) instead of (\ref{12-6-4-2}).
Then, we obtain a slightly stronger evaluation than Theorem \ref{t8-27-1}.

\begin{thm}\Label{t8-27-2}
Assume that 
$Q_E$ is a normalized distribution on ${\cal E}$,
$P_{A,E}$ is a sub-distribution on $\cA \times \cE$,
and
a linear random hash function $f_{\bX}$ from $\cA$ to $\cM=\{1, \ldots, \sM\}$
is $\varepsilon$-almost dual universal$_2$.
Then, the random hash function $f_{\bX}$ satisfies
\begin{align}
\rE_{\bX} I'(f_{\bX}(A)|E|P_{A,E} ) 
\le 
\underline{\Delta}_{I,\min}(\sM,\varepsilon|P_{A,E}) ,
\Label{12-5-2-d}
\end{align}
where
\begin{align}
\underline{\Delta}_{I,\min}(\sM,\varepsilon|P_{A,E}) 
&:=\min_{Q_E}
\min_{P_{A,E}': P_{E}'\le Q_E,}
\eta( \|P_{A,E}-P_{A,E}'\|_1  ,\log \sM)
+ \log (1+\varepsilon M e^{-{H}_{\min}(A|E|P_{A,E}' \|P_E )} )
\Label{12-5-2-f}
\\
&=
\min_{\epsilon_1 >0}
\eta( \epsilon_1 ,\log \sM ) 
+\log (1+\varepsilon \sM e^{-{H}_{\min}^{\downarrow,\epsilon_1}(A|E|P_{A,E})} )
\Label{12-5-2-c} \\
&= \min_{R'} 
\eta( 
\min_{P_{A,E}': P_{E}'\le P_E, {H}_{\min}(A|E|P_{A,E}' \|P_E )\ge R} 
\|P_{A,E}-P_{A,E}'\|_1  
 ,\log \sM ) 
+\log (1+\varepsilon \sM e^{-R'}) .
\Label{12-5-2-x} 
\end{align}
\end{thm}

Further, by using similar discussions as Sections \ref{cqs4-5} and \ref{s4-1},
the upper bound 
$\underline{\Delta}_{I,\min}(\sM,\varepsilon| P_{A,E}|P_{A,E}) $
can be evaluated as follows.

\begin{thm}\Label{t8-27-3}
\begin{align}
\underline{\Delta}_{I,\min}(\sM,\varepsilon| P_{A,E}|P_{A,E}) 
\le &
\min_{R'}
\eta( 
P_{A,E}\{ (a,e)| P_{A|E}(a|e) > e^{-R'} \}, \log \sM)
+\log (1+\varepsilon \sM e^{-R'}) \Label{8-28-1b}
\\
\le &
\min_{R'}
\eta( \min_{s \ge 0} e^{s (R'-H_{1+s}^{\downarrow}(A|E|P_{A,E}) } , \log \sM)
+\log (1+\varepsilon \sM e^{-R'})
\Label{8-28-1}
\end{align}
\end{thm}

\begin{proof}
Inequality (\ref{8-28-1b}) follows from 
Lemma \ref{L8-29-1} and (\ref{12-5-2-x}).
Inequality (\ref{8-28-1}) follows from (\ref{8-29-3b}) with $Q_E=P_E$.
\end{proof}

Now, we consider the asymptotic behavior of $\underline{\Delta}_{I,\min}(\lceil e^{nR} \rceil,\varepsilon| P_{A,E}^n) $.
\begin{thm}\Label{t8-27-4}
Any polynomial $P(n)$ satisfies 
\begin{align}
\lim_{n \to \infty}
\frac{1}{n}
\underline{\Delta}_{I,\min}(\lceil e^{nR} \rceil,P(n) | P_{A,E}^n) 
=
R-H(A|E|P_{A,E})
\Label{8-28-15}
\end{align}
for $R \ge H(A|E|P_{A,E})$.
\end{thm}

Theorem \ref{t8-27-4} shows that 
$\varepsilon$-almost dual universal$_2$ hash functions realize the 
asymptotically optimal performance 
in the sense of equivocation rate.
Further, 
Theorem \ref{t8-27-4} clarifies that 
the smoothing bound of min entropy yields 
the optimal evaluation 
in the sense of equivocation rate.

\begin{proof}
It is known by \cite{Wyner} that 
any sequence of hash function from ${\cal A}$ to 
$\{1, \ldots, \lceil e^{nR} \rceil \}$
satisfies
\begin{align}
\liminf_{n \to \infty}\frac{1}{n}\rE_{\bX,n} I'(f_{\bX,n}(A)|E|P_{A,E} ) 
\ge R-H(A|E|P_{A,E}).
\end{align}
Hence, it is enough to show that
\begin{align}
\limsup_{n \to \infty}
\frac{1}{n}
\underline{\Delta}_{I,\min}(\lceil e^{nR} \rceil,P(n) | P_{A,E}^n) 
\le
R-H(A|E|P_{A,E}). \Label{8-28-4}
\end{align}

We choose $R' < H(A|E|P_{A,E})$.
Relation (\ref{8-28-1}) implies that
\begin{align}
\frac{1}{n}
\underline{\Delta}_{I,\min}(\lceil e^{nR} \rceil,P(n) | P_{A,E}^n) 
\le &
\frac{1}{n}
\eta( \min_{s \ge 0} e^{s n(R'-H_{1+s}^{\downarrow}(A|E|P_{A,E}) } , nR)
+\frac{1}{n}
\log (1+P(n)  e^{n(R-R')}) 
\end{align}
Since $R' < H(A|E|P_{A,E})$,
the value $\min_{s \ge 0} e^{s n(R'-H_{1+s}^{\downarrow}(A|E|P_{A,E}) }$ goes to zero exponentially.
Hence, the term \par
\noindent $\frac{1}{n}
\eta( \min_{s \ge 0} e^{s n(R'-H_{1+s}^{\downarrow}(A|E|P_{A,E}) } , nR)$
goes to zero.
Since $\frac{1}{n}
\log (1+P(n) e^{n(R-R')}) 
\le R-R'+\frac{1}{n}\log (1+P(n) )
\to R-R'$, we have
\begin{align}
\limsup_{n \to \infty}
\frac{1}{n}
\underline{\Delta}_{I,\min}(\lceil e^{nR} \rceil,P(n) | P_{A,E}^n) 
\le R-R'.
\end{align}
Since $R'$ is an arbitrary real number satisfying $R'< H(A|E|P_{A,E})$,
we obtain (\ref{8-28-4}).
\end{proof}

\section{Conclusion}
We have derived upper bounds for the leaked information in the modified mutual information
criterion and the $L_1$ distinguishability criterion
when we apply 
an $\varepsilon$-almost dual universal$_2$ hash function for privacy amplification. 
(Theorems \ref{Lem12} and \ref{Lem11} in Section \ref{s4-1}). 
Then, we have derived lower bounds on their exponential decreasing rates in the i.i.d. setting.
(Theorem \ref{t3-16-1} in Section \ref{s4-1-b}). 

We have rigorously compared the exponents by the smoothing bound of min-entropy 
and R\'{e}nyi entropy of order 2.
That is, we have clarified the upper bounds of leaked information via 
the smoothing of min-entropy in the both criteria.
That is, we have compared 
$\Delta_{d,2}(M,\varepsilon|P_{A,E})$ and 
$\Delta_{d,\min}(M,\varepsilon|P_{A,E})$
for R\'{e}nyi entropy of order 2,
and have done 
$\Delta_{I,2}(M,\varepsilon|P_{A,E})$ and
$\Delta_{I,\min}(M,\varepsilon|P_{A,E})$
for modified mutual information criterion.
We have derived the exponents of the upper bounds 
(Theorem \ref{L3-18-3} in Section \ref{s4-1}),
and have shown that the exponents are strictly worse than the exponents by  
the smoothing  bound of R\'{e}nyi entropy of order 2 
(Lemma \ref{L3-19-11} in Section \ref{s4-1}).
This fact shows the importance of the smoothing of R\'{e}nyi entropy of order 2.
The obtained exponents are summarized in Table \ref{table2}.


Due to Pinsker inequality and Inequality (\ref{8-26-9}), 
the exponential convergence of one criterion yields the exponential convergence of the other criterion.
However, we have shown that
better exponential decreasing rates can be obtained by separate derivations. 
For example,
the smoothing of R\'{e}nyi entropy of order 2 yields 
the exponent $e_{d}(P_{A,E}|R)$ for the $L_1$ distinguishability criterion,
which yields the exponent $e_{d}(P_{A,E}|R)$ for the modified mutual information criterion by using Pinsker inequality.
Similarly, the smoothing of R\'{e}nyi entropy of order 2 yields 
the exponent $e_{I}(P_{A,E}|R)$ for the modified mutual information criterion,
which yields the exponent $\frac{e_{I}(P_{A,E}|R)}{2}$ for the $L_1$ distinguishability criterion by Inequality (\ref{8-26-9}).
Since $e_{d}(P_{A,E}|R) \ge \frac{e_{I}(P_{A,E}|R)}{2}$ and $e_{I}(P_{A,E}|R) \ge e_{d}(P_{A,E}|R)$,
the exponents directly derived by the smoothing of R\'{e}nyi entropy of order 2 
are better than the exponents derived from the combination of the exponent for the other criterion and the inequality.

\begin{table}[htb]
  \caption{Summary of obtained lower bounds on exponents.}
\begin{center}
  \begin{tabular}{|l|c|c|} \hline
{Method}   & $L_1$  & MMI \\ \hline
{smooth R\'{e}nyi 2 } & {$e_{d}(P_{A,E}|R)$}  & {$e_{I}(P_{A,E}|R)$} \\
 \hline
smooth min & {$\tilde{e}_{d}(P_{A,E}|R)$}  & {$\tilde{e}_{I}(P_{A,E}|R)$} 
\\
\hline
  \end{tabular}
\end{center}

\vspace{2ex}
smooth R\'{e}nyi 2 is 
the exponent for privacy amplification via the smoothing of R\'{e}nyi entropy of order 2.
smooth min is 
the exponent for privacy amplification via the smoothing of min entropy.
L2 is the $L_1$ distinguishability criterion.
MMI is the modified mutual information criterion. 
\Label{table2}
\end{table}

We have also shown that 
the application of 
$\varepsilon$-almost dual universal hash function 
attains the asymptotically optimal performance in the sense of  
the second order asymptotics as well as in that of the asymptotic equivocation rate.
These facts have been shown by using the smoothing of min entropy.
We can conclude that
$\varepsilon$-almost dual universal hash functions
are very a useful class of hash functions.
Further, these discussions show that
the smoothing of min entropy
is sufficiently powerful except for 
the exponential decreasing rate.
That is, the exponential decreasing rate
requires more delicate evaluation than other settings.

\section*{Acknowledgments}
The author is grateful to Dr. Toyohiro Tsurumaru,
Dr. Shun Watanabe, 
Dr. Marco Tomamichel,
Dr. Mario Berta,
Dr. William Henry Rosgen,
Dr. Li Ke,
and Dr. Markus Grassl
 for a helpful comments.
He is also grateful to the
referee of the first version of \cite{Tsuru}
for informing the literatures \cite{DS05,FS08}.
He is partially supported by a MEXT Grant-in-Aid for Scientific Research (A) No. 23246071
and the National Institute of Information and Communication Technology (NICT), Japan.
The Centre for Quantum Technologies is funded by the
Singapore Ministry of Education and the National Research Foundation
as part of the Research Centres of Excellence programme.

\appendices

\section{Proof of Lemma \ref{cor1}}\Label{scor1}
For 
two non-negative functions $X(e)$ and $Y(e)$, 
the reverse H\"{o}lder inequality \cite{Kuptsov}
\begin{align*}
\sum_e X(e) Y(e) \ge
(\sum_e X(e)^{1/(1+s)})^{1+s}
(\sum_e Y(e)^{-1/s})^{-s}
\end{align*}
holds for $s\in (0,\infty]$.
Substituting 
$\sum_a P_{A,E}(a,e)^{1+s}$
and $Q_E(e)^{-s}$
to $X(e)$ and $Y(e)$, we obtain  
\begin{align*}
& e^{-s H_{1+s}(A|E|P_{A,E}\| Q_E )} \\
= & 
\sum_e 
\sum_a 
P_{A,E}(a,e)^{1+s}
Q_E(e)^{-s} \\
\ge &
(\sum_e 
(\sum_a P_{A,E}(a,e)^{1+s})^{1/(1+s)})^{1+s}
(\sum_e Q_E(e)^{-s\cdot -1/s})^{-s} \\
= &
(\sum_e 
(\sum_a P_{A,E}(a,e)^{1+s})^{1/(1+s)})^{1+s} \\
= & 
(\sum_{e} (\sum_{a} P_{A,E}(a,e)^{1+s})^{\frac{1}{1+s}} )^{1+s}
\end{align*}
for $s\in (0,\infty]$.
Since the equality holds
when $Q_E(e)=
(\sum_{a} P_{A,E}(a,e)^{1+s})^{1/(1+s)} /
\sum_e (\sum_{a} P_{A,E}(a,e)^{1+s})^{1/(1+s)}$,
we obtain
\begin{align*}
e^{-s H_{1+s}^{\uparrow}(A|E|P_{A,E})},
=\min_{Q_E}e^{-s H_{1+s}(A|E|P_{A,E}\| Q_E )} 
=
(\sum_{e} (\sum_{a} P_{A,E}(a,e)^{1+s})^{\frac{1}{1+s}} )^{1+s}
\end{align*}
which implies (\ref{8-26-8-c}) with $s\in (0,\infty]$.

For 
two non-negative functions $X(e)$ and $Y(e)$, 
the H\"{o}lder inequality 
\begin{align*}
\sum_e X(e) Y(e) \le
(\sum_e X(e)^{1/(1+s)})^{1+s}
(\sum_e Y(e)^{-1/s})^{-s}
\end{align*}
holds for $s\in [-1,0)$.
The same substitution yields
\begin{align*}
e^{-s H_{1+s}(A|E|P_{A,E}\| Q_E )} 
\le  
(\sum_{e} (\sum_{a} P_{A,E}(a,e)^{1+s})^{\frac{1}{1+s}} )^{1+s}
\end{align*}
for $s\in [-1,0)$.
Hence, similarly
we obtain (\ref{8-26-8-c}) with $s\in [-1,0)$.

\section{Proof of Lemma \ref{cor}}\Label{scor}
For $s \in (0,1]$ and
two functions $X(a)$ and $Y(a)$, 
the H\"{o}lder inequality
\begin{align*}
\sum_a X(a) Y(a) \le
(\sum_a |X(a)|^{1/(1-s)})^{1-s}
(\sum_a |Y(a)|^{1/s})^{s}
\end{align*}
holds.
The equality holds only when 
$X(a)$ is a constant times of $Y(a)$.
Substituting 
$P_{A,E}(a,e)$
and $(\frac{P_{A,E}(a,e)}{P_{E}(e)})^s$
to $X(a)$ and $Y(a)$, we obtain  
\begin{align*}
& e^{-s H_{1+s}^{\downarrow}(A|E|P_{A,E})} \\
= & 
\sum_e 
\sum_a 
P_{A,E}(a,e)
(\frac{P_{A,E}(a,e)}{P_{E}(e)})^s \\
\le &
\sum_e 
(\sum_a P_{A,E}(a,e)^{1/(1-s)})^{1-s}
(\sum_a \frac{P_{A,E}(a,e)}{P_{E,\normal}(e)})^s \\
= &
\sum_e 
(\sum_a P_{A,E}(a,e)^{1/(1-s)})^{1-s} \\
= &
e^{-s H_{\frac{1}{1-s}}^{\uparrow}(A|E|P_{A,E} )}
\end{align*}
for $s \in  (0,1]$
because 
$\sum_a \frac{P_{A,E}(a,e)}{P_{E,\normal}(e)}
=\frac{P_{E}(e)}{P_{E,\normal}(e)}
\le 1$.
The equality condition holds only when 
$P_{A|E=e}$ is uniform distribution for all $e \in {\cal E}$.

For $s \in [-1,0)$ and
two functions $X(a)$ and $Y(a)$, 
the reverse H\"{o}lder inequality \cite{Kuptsov}
\begin{align*}
\sum_a X(a) Y(a) \ge
(\sum_a |X(a)|^{1/(1-s)})^{1-s}
(\sum_a |Y(a)|^{1/s})^{s}
\end{align*}
holds.
The same substitution yields
\begin{align*}
e^{-s H_{1+s}^{\downarrow}(A|E|P_{A,E})} 
\ge 
e^{-s H_{\frac{1}{1-s}}^{\uparrow}(A|E|P_{A,E} )}
\end{align*}
for $s \in [-1,0)$
because
$(\sum_a \frac{P_{A,E}(a,e)}{P_{E,\normal}(e)})^s
=(\frac{P_{E}(e)}{P_{E,\normal}(e)})^s \ge 1$.
The equality condition holds only when 
$P_{A|E=e}$ is uniform distribution for all $e \in {\cal E}$.

\section{Proof of Lemma \ref{L7-1}}\Label{pL7-1}
First, we show \eqref{1-5-1}.
Taking the limit $s\to 0$, we obtain
\begin{align}
&H(A|E|P_{A,E})
=
-\frac{d \phi(s|A|E|P_{A,E})}{ds}|_{s=0}\nonumber \\
=&
-\lim_{s\to 0}
\frac{\phi(s|A|E|P_{A,E})}{s}
=
\lim_{s\to 0}
H_{1+s}^{\uparrow}(A|E|P_{A,E}).
\end{align}
The remaining properties are shown by the following lemma.

\begin{lem}\Label{L3-19-1}
\begin{align}
&-\frac{d}{ds}s H_{1+s}^{\uparrow}(A|E|P_{A,E}) \nonumber\\
=& 
\sum_{a,e}P_{A,E;s}(a,e)
\Bigl(
 \log P_{A|E}(a|e)- \frac{1}{1+s}\log 
(\sum_a P_{A|E}(a|e)^{1+s} )
\Bigr)
+ \phi (\frac{s}{1+s}| A|E|P_{A,E}) ,\Label{3-19-1}\\
&-\frac{d^2}{ds^2}s H_{1+s}^{\uparrow}(A|E|P_{A,E}) 
\nonumber\\
=& 
(1+s)
\sum_{a,e}P_{A,E;s}(a,e)
\Bigl(
\frac{1}{1+s} \log P_{A|E}(a|e)- \frac{1}{(1+s)^2}\log 
(\sum_a P_{A|E}(a|e)^{1+s} )
\Bigr)^2\nonumber\\
&-
(1+s)
\Bigl(\sum_{a,e}P_{A,E;s}(a,e)
\Bigl(
\frac{1}{1+s} \log P_{A|E}(a|e)- \frac{1}{(1+s)^2}\log 
(\sum_a P_{A|E}(a|e)^{1+s} )
\Bigr)\Bigr)^2. \Label{3-19-2}
\end{align}
\end{lem}

Hence, 
when we regard $H_{1}^{\uparrow}(A|E|P_{A,E})$ as $H(A|E|P_{A,E})$
and
$P_{A|E=e}$ is not a uniform distribution for an element $e \in {\cal E}$,
the function 
$s \mapsto -s H_{1+s}^{\uparrow}(A|E|P_{A,E}) $
is strictly convex in $(-1,\infty)$.
That is,
the map $s \mapsto s H_{1+s}^{\uparrow}(A|E|P_{A,E})$ is strictly concave and then
the map $s \mapsto H_{1+s}^{\uparrow}(A|E|P_{A,E})$ 
is strictly monotonically decreasing for $s \in (-1,\infty)$.

\begin{proof}
We define
\begin{align*}
\varphi(s):=
\sum_{e}P_E(e)(\sum_a P_{A|E}(a|e)^{1+s})^{\frac{1}{1+s}}.
\end{align*}
Then,
\begin{align*}
&\frac{d \varphi(s)}{ds} \\
=&
\sum_{a,e}
\frac{P_{A|E}(a|e)^{1+s} P_E(e)}{(\sum_a P_{A|E}(a|e)^{1+s})^{\frac{s}{1+s}} (\sum_{e} P_E(e) }
\Bigl(\frac{1}{1+s} \log P_{A|E}(a|e)- \frac{1}{(1+s)^2}\log 
(\sum_a P_{A|E}(a|e)^{1+s} ) \Bigr)\\
=&
\varphi(s)
\sum_{a,e}
P_{A,E;s}(a,e)
\Bigl(\frac{1}{1+s} \log P_{A|E}(a|e)- \frac{1}{(1+s)^2}\log 
(\sum_a P_{A|E}(a|e)^{1+s} ) \Bigr).
\end{align*}
Since
\begin{align*}
&-\frac{d}{ds}s H_{1+s}^{\uparrow}(A|E|P_{A,E}) \\
=& 
\phi (\frac{s}{1+s}| A|E|P_{A,E}) +
(1+s)\frac{d \varphi(s)}{ds}\varphi(s)^{-1},
\end{align*}
we obtain (\ref{3-19-1}).

Next, we show (\ref{3-19-2}).
Since
\begin{align*}
&\frac{d^2 \varphi(s)}{ds^2}\\
=&
\sum_{a,e}
\frac{P_{A|E}(a|e)^{1+s} P_E(e)}{(\sum_a P_{A|E}(a|e)^{1+s})^{\frac{s}{1+s}} (\sum_{e} P_E(e) }
\Bigl(\frac{1}{1+s} \log P_{A|E}(a|e)- \frac{1}{(1+s)^2}\log 
(\sum_a P_{A|E}(a|e)^{1+s} ) \Bigr)^2\\
&+
\sum_{a,e}
\frac{P_{A|E}(a|e)^{1+s} P_E(e)}{(\sum_a P_{A|E}(a|e)^{1+s})^{\frac{s}{1+s}} (\sum_{e} P_E(e) }
\Bigl(-\frac{2}{(1+s)^2} \log P_{A|E}(a|e)+ \frac{2}{(1+s)^3}\log 
(\sum_a P_{A|E}(a|e)^{1+s} ) \Bigr) \\
=&
\varphi(s)
\sum_{a,e}
P_{A,E;s}(a,e)
\Bigl(\frac{1}{1+s} \log P_{A|E}(a|e)- \frac{1}{(1+s)^2}\log 
(\sum_a P_{A|E}(a|e)^{1+s} ) \Bigr)^2\\
&-\frac{2}{(1+s)} 
\frac{d \varphi(s)}{ds},
\end{align*}
we have
\begin{align*}
&\frac{d^2}{ds^2}(1+s) \phi (\frac{s}{1+s}| A|E|P_{A,E}) \\
=& 
(1+s)\frac{d^2}{ds^2} \phi (\frac{s}{1+s}| A|E|P_{A,E}) 
+2 \frac{d}{ds} \phi (\frac{s}{1+s}| A|E|P_{A,E}) \\
=&
(1+s)
\frac{\varphi(s) \frac{d^2 \varphi(s)}{ds^2}-
\frac{d \varphi(s)}{ds}^2}{\varphi(s)^2}
+2\frac{\frac{d \varphi(s)}{ds}}{\varphi(s)} \\
=&
(1+s)
\frac{\varphi(s) \frac{d^2 \varphi(s)}{ds^2}-
\frac{d \varphi(s)}{ds}^2}{\varphi(s)^2}
+2\frac{\frac{d \varphi(s)}{ds}}{\varphi(s)} \\
=&
(1+s)
\sum_{a,e}
P_{A,E;s}(a,e)
\Bigl(\frac{1}{1+s} \log P_{A|E}(a|e)- \frac{1}{(1+s)^2}\log 
(\sum_a P_{A|E}(a|e)^{1+s} ) \Bigr)^2\\
&-(1+s)
\Bigl( \sum_{a,e}
P_{A,E;s}(a,e)
\Bigl(\frac{1}{1+s} \log P_{A|E}(a|e)- \frac{1}{(1+s)^2}\log 
(\sum_a P_{A|E}(a|e)^{1+s} ) \Bigr)\Bigr)^2,
\end{align*}
which implies (\ref{3-19-2}).
\end{proof}

\section{Proof of Theorem \ref{l8-24-1}}\Label{s8-24}
First, we show that the modified mutual information criterion $I'(A|E|P)= \log |\cA| -H(A|E|P)$ satisfies all of the above conditions.
We can trivially check the conditions {\bf C4} Ideal case and {\bf C5} Normalization.
We show other conditions.
{\bf C1} Chain rule can be shown as follows.
\begin{align*}
&I'(A,B|E|P)
=\log |\cA|+\log |\cB|
- H(A,B,E|P) + H(E|P) \\
=&\log |\cA|+\log |\cB|
- H(B,E|P) + H(E|P)
- H(A,B,E|P) + H(B,E|P) \\
=& \log |\cA|+\log |\cB|
- H(B|E|P)
- H(A|B,E|P)
= I'(A|B,E|P) +I'(B|E|P).
\end{align*}

When two marginal distributions $P_{E,1}$ and $P_{E,2}$ are distinghuishable on ${\cal E}$,
\begin{align*}
& I'(A|E|\lambda P_1+(1-\lambda) P_2)
=\log |\cA|
- H(A,E|\lambda P_1+(1-\lambda) P_2) + H(E|\lambda P_1+(1-\lambda) P_2) \\
=&\log |\cA|
- \lambda  H(A,E|P_1)
-(1-\lambda) H(A,E|P_2) 
- h(\lambda)
+ \lambda  H(E|P_1)
+(1-\lambda) H(E|P_2) 
+ h(\lambda) \\
=&\log |\cA|
- \lambda  H(A,E|P_1)
-(1-\lambda) H(A,E|P_2) 
+ \lambda  H(E|P_1)
+(1-\lambda) H(E|P_2) \\
=&\lambda I'(A|E|P_1)+(1-\lambda)I'(A|E|P_2),
\end{align*}
which implies {\bf C2} Linearity.
$I'(A|E|P)=D(P\|P_{\mix,\cA}\otimes P_E)\ge 0$.
Since $H(A,E|P)\ge 0 $, $I'(A|E|P)$ satisfies {\bf C3} Range.
Thus, $I'(A|E|P)$ satisfies all of the above properties.

Next, we show that 
an quantity satisfying all of the above properties 
is the modified mutual information criterion $I'(A|E|P)= \log |\cA| -H(A|E|P)$.
For this purpose, 
we focus on $\tilde{H}(A|E|P):=\log |\cA| - C(A|E|P)$.
Due to {\bf C1} Linearity,
we have
\begin{align*}
\tilde{H}(A|E|P)=
\sum_e P_E(e) \tilde{H}(A|E|P_{A|E=e}).
\end{align*}
Further, we see that the quantity $\tilde{H}(A|E|P_{A|E=e})$ satisfies 
Khinchin's axioms \cite{Khinchin} for entropy
because of the remaining properties.
Hence, we find that $\tilde{H}(A|E|P_{A|E=e})=H(P_{A|E=e})$.
Thus, $\tilde{H}(A|E|P)$ is equal to the conditional entropy ${H}(A|E|P)$.
Hence, $C(A|E|P)=I'(A|E|P)$.

\section{Proof of Proposition \ref{Lem6-3}}\Label{pfLem6-1}
Since the proof of Proposition \ref{Lem6-3} 
is related to $\delta$-biased ensemble,
we make several preparations before starting the proof of Proposition \ref{Lem6-3}.
According to Dodis and Smith\cite{DS05},
we introduce $\delta$-biased ensemble of random variables $W_{\bX}$
on a vector space over a general finite field $\FF_q$,
where $q$ is the power of the prime $p$.
First, we fix a non-degenerate bilinear form 
$(~,~)$ from $\FF_q^2$ to $\FF_p$.
Then, we define $(x\cdot y) \in \FF_p$ for $x,y\in \FF_q^n$
as $(x\cdot y):=\sum_{j=1}^n x_j \cdot y_j$.
For a given $\delta>0$,
an ensemble of random variables $\{W_{\bX}\}$ on $\FF_q^n$
is called {\it $\delta$-biased}
when the inequality
\begin{align}
\rE_{\bX} |\rE_{W_{\bX}} \omega_p^{(x\cdot W_{\bX})}|^2 \le \delta^2
\Label{12-27-9}
\end{align}
holds for any $x\neq 0 \in \FF_q^n$,
where $\omega_p:= e^{\frac{2\pi i}{p}}$.

We denote the random variable subject to the uniform distribution on a code $C\in \FF_q^n$,
by $W_C$.
Then,
\begin{align}
\rE_{W_C} \omega_p^{(x\cdot W_{C})}
=
\left\{
\begin{array}{ll}
0 & \hbox{ if } x \notin C^{\perp} \\
1 & \hbox{ if } x \in C^{\perp} .
\end{array}
\right. \Label{3-23-1}
\end{align}
Using the above relation, 
as is suggested in \cite[Case 2]{DS05},
we obtain the following lemma. 

\begin{lem}\Label{Lem6-0}
When
a random code $C_{\bX}$ in $\FF_q^n$ 
is $\varepsilon$-almost dual universal
with minimum dimension $t$,
the ensemble of random variables $W_{C_{\bX}}$ in $\FF_q^n$
is $\sqrt{\varepsilon q^{-t}}$-biased.
\end{lem}

\begin{proof}
$C^{\perp}_{\bX}$ is $\varepsilon$-almost universal
with maximum dimension $n-t$ in $\FF_q^n$.
Hence, for any $x \in \FF_q^n$,
the probability $\Pr \{x \in C^{\perp}_{\bX}\}$ is less than
$\varepsilon q^{-t}$.
Thus, (\ref{3-23-1}) guarantees that
the ensemble of random variables $W_{C_{\bX}}$ in $\FF_q^n$
is $\sqrt{\varepsilon q^{-t}}$-biased.
\end{proof}

In the following, we treat the case of ${\cal A}=\FF_q^n$.
Given
a joint sub-distribution $P_{A,E}$ on ${\cal A} \times {\cal E}$
and a normalized distribution $P_{W}$ on ${\cal A}$,
we define another joint sub-distribution 
$P_{A,E} * P_{W}(a,e):=
\sum_{w} P_{W}(w) P_{A,E}(a-w,e)$.
Using these concepts,
Dodis and Smith\cite{DS05}
evaluated the average
of $ {d_{2}}(A |E|P_{A,E} * P_{W_{\bX}} \|Q_E )$ as follows.

\begin{proposition}[{\cite[Lemma 4]{DS05}}]\Label{Lem6-1}
For any joint sub-distribution $P_{A,E}$ on ${\cal A} \times {\cal E}$
and any normalized distribution $Q_E$ on ${\cal E}$,
a $\delta$-biased
ensemble of random variables $\{W_{\bX}\}$ on ${\cal A}=\FF_q^n$
satisfies
\begin{align}
\rE_{\bX} {d_{2}}(A |E|P_{A,E} * P_{W_{\bX}} \|Q_E )
\le
\delta^2
e^{-H_{2}(A|E|P_{A,E} \| Q_E)}.\Label{Lem6-1-eq1}
\end{align}
More precisely,
\begin{align}
& \rE_{\bX} {d_{2}}(A |E|P_{A,E} * P_{W_{\bX}} \|Q_E ) 
\nonumber \\
\le &
\delta^2
{d_{2}}(A |E|P_{A,E} \|Q_E ). \Label{Lem6-1-eq2}
\end{align}
\end{proposition}

The original proof by Dodis and Smith\cite{DS05}
discussed in the case with $q=2$.
Fehr and Schaffner \cite{FS08} extended this lemma to the quantum setting
in the case with $q=2$.
Their proof is based on Fourier analysis and easy to understand.
The proof with a general prime power $q$ is given latter.
by generalizing the idea by Fehr and Schaffner \cite{FS08}.
Dodis and Smith\cite[Lemma 6]{DS05} also considered the case with a general prime power $q$.
They did not explicitly give Proposition \ref{Lem6-1} 
and the definition (\ref{12-27-9})
with a general prime power $q$.

Proposition \ref{Lem6-3} essentially coincides with Proposition \ref{Lem6-1}.
However, 
the concept ``$\delta$-biased" does not concern a linear random hash function
while the concept ``$\varepsilon$-almost dual universality$_2$" does it
because the former is defined for the ensemble of random variables.
That is, 
the latter is a generalization of a universal$_2$ linear hash function while the former does not.
Hence,
Proposition \ref{Lem6-1} cannot directly provide the performance of a linear random hash function.
In contrast, Proposition \ref{Lem6-3} gives how the privacy amplification by 
a linear hash function decreases the leaked information. 
Therefore, in the main part of this paper,
using Proposition \ref{Lem6-3}, 
we treat the exponential decreasing rate when we apply 
the privacy amplification by 
an $\varepsilon$-almost dual universal$_2$ linear hash function.

\begin{proofof}{Proposition \ref{Lem6-3}}
Due to Lemma \ref{Lem6-0} and (\ref{Lem6-1-eq1}), we obtain
\begin{align}
\rE_{\bX} {d_{2}}(A |E|P_{A,E} * P_{W_{C_{\bX}}} \| Q_E )
\le
\varepsilon q^{-t}
e^{-H_{2}(A|E|P_{A,E} \|Q_E)}.
\Label{12-18-1}
\end{align}

Denoting the quotient class with respect to the subspace $C$ with the representative $a\in \cA$
by $[a]$, we obtain
\begin{align*}
P_{A,E} * P_{W_{C}} (a,e)
=&
\sum_{w \in C}
q^{-t}
P_{A,E}(a-w,e) \\
=&
q^{-t}
P_{A,E}([a],e) .
\end{align*}
Now, we focus on the relation 
${\cal A} \cong {\cal A}/C \times C\cong f_C({\cal A}) \times C$.
Then,
\begin{align*}
P_{A,E} * P_{W_{C_{\bX}}} (b,w,e)
=
q^{-t}
P_{f_C(A),E}(b,e).
\end{align*}
Thus,
\begin{align}
& {d_{2}}(A |E|P_{A,E} * P_{W_{C}}  \|Q_E )\nonumber  \\
= &
q^{-t}
{d_{2}}(f_{C}(A) |E|P_{f_{C}(A),E}\| Q_E ) \nonumber \\
=&
q^{-t}
{d_{2}}(f_{C}(A) |E|P_{A,E}\| Q_E).
\Label{3-23-2}
\end{align}
Therefore, (\ref{12-18-1}) implies
\begin{align*}
& \rE_{\bX} q^{-t}
{d_{2}}(f_{C_{\bX}}(A) |E| P_{A,E}\| Q_E) \\
\le &
\varepsilon q^{-t}
e^{-H_{2}(A|E|P_{A,E}\| Q_E)},
\end{align*}
which implies (\ref{12-5-9}).

Similarly, Lemma \ref{Lem6-0}, (\ref{Lem6-1-eq2}), and (\ref{3-23-2}) imply that
\begin{align*}
& \rE_{\bX} q^{-t}
{d_{2}}(f_{C_{\bX}}(A) |E| P_{A,E}\| Q_E) \\
\le &
\varepsilon q^{-t} e^{-H_{2}(A|E|P_{A,E}\| Q_E)}.
\end{align*}
Since $\rE_{\bX}
{d_{2}}(f_{C_{\bX}}(A) |E| P_{A,E}\| Q_E) 
=
\rE_{\bX} 
e^{-H_{2}(f_{C_{\bX}}(A)|E|P_{A,E} \| Q_E)}
-\frac{1}{q^{n-t}} e^{D_2(P_{E} \| Q_E)}$,
we have (\ref{Lem6-3-eq2}).
\end{proofof}

To start our proof of Proposition \ref{Lem6-1},
we make preparation before our proof of Proposition \ref{Lem6-1}.
First, remember that ${\cal A}$ is a vector space $\FF_q^n$ and 
${\cal E}$ is a general discrete set.
We define the $\ell^2$ norm over the space $L^2({\cal A}\times {\cal E})$
as
\begin{align}
\|f\|_2^2:= \sum_{a \in {\cal A},e\in {\cal E}} |f(a,e)|^2,
\quad \forall f \in L^2({\cal A}\times {\cal E}).
\end{align}
Then, we define the discrete Fourier transform ${\cal F}$ on
$L^2({\cal A}\times {\cal E})$ as
\begin{align}
{\cal F}(f)(a',e):= 
q^{-\frac{n}{2}} \sum_{a \in {\cal A}} 
\omega_p^{(a'\cdot a)}
f(a,e),
\quad \forall f \in L^2({\cal A}\times {\cal E}),
\forall a' \in {\cal A},
\forall e \in {\cal E},
\end{align}
which satisfies $\|{\cal F}f\|_2=\|f\|_2$.
For $\forall f,g \in L^2({\cal A}\times {\cal E})$,
the convolution $f*g$:
\begin{align}
f*g(a,e):=\sum_{a' \in {\cal A}}f(a-a',e)g(a',e).
\end{align}
satisfies 
\begin{align}
{\cal F}(f*g)(a,e)=q^{\frac{n}{2}} {\cal F}(f)(a,e) {\cal F}(g)(a,e).
\Label{12-27-6}
\end{align}

We prepare the following lemma.
\begin{lem}
When $f_{P_{A,E},Q_E} \in L^2({\cal A}\times {\cal E})$ is defined as
\begin{align}
f_{P_{A,E},Q_E}(a,e):= P_{A,E}(a,e)Q_E(e)^{-\frac{1}{2}} ,
\end{align}
we have
\begin{align}
\| f_{P_{A,E},Q_E} \|_2^2
&= e^{-H_2(A|E|P_{A,E}\| Q_E) } \Label{12-27-3}
\\
\sum_{e\in {\cal E}}|{\cal F}(f_{P_{A,E},Q_E})(0,e)|^2
&= e^{D_2(P_{E}\| Q_E) } \Label{12-27-4}
\\
\sum_{a\neq 0 \in {\cal A} e\in {\cal E}}
|{\cal F}(f_{P_{A,E},Q_E})(a,e)|^2 
&= d_2(A|E|P_{A,E}\| Q_E)  .\Label{12-27-5}
\end{align}
\end{lem}
\begin{proof}
(\ref{12-27-3}) and 
(\ref{12-27-4}) are shown as follows.
\begin{align*}
\| f_{P_{A,E},Q_E} \|_2^2
&=
\sum_{a,e} (P_{A,E}(a,e)Q_E(e)^{-\frac{1}{2}})^2
= e^{-H_2(A|E|P_{A,E}\| Q_E) } \\
\sum_{e\in {\cal E}}|{\cal F}(f_{P_{A,E},Q_E})(0,e)|^2
&= 
\sum_{e} (\sum_{a}P_{A,E}(a,e)Q_E(e)^{-\frac{1}{2}})^2 
= 
\sum_{e} (P_{E}(e)Q_E(e)^{-\frac{1}{2}})^2 
= e^{D_2(P_{E}\| Q_E) } .
\end{align*}
(\ref{12-27-5}) is shown as follows.
\begin{align*}
& \sum_{a\neq 0 \in {\cal A}, e\in {\cal E}}
|{\cal F}(f_{P_{A,E},Q_E})(a,e)|^2
= 
\|{\cal F}(f_{P_{A,E},Q_E})\|_2^2
-
\sum_{ e\in {\cal E}} |{\cal F}(f_{P_{A,E},Q_E})(0,e)|^2 \\
=& 
\| f_{P_{A,E},Q_E}\|_2^2
-
\sum_{ e\in {\cal E}} |{\cal F}(f_{P_{A,E},Q_E})(0,e)|^2 \\
=& 
e^{-H_2(A|E|P_{A,E}\| Q_E) }
- e^{D_2(P_{E}\| Q_E) } 
=
d_2(A|E|P_{A,E}\| Q_E)  .
\end{align*}
\end{proof}

\begin{proofof}{Proposition \ref{Lem6-1}}
Now, we choose $g_{\bX} \in L^2({\cal A}\times {\cal E})$ as
\begin{align}
g_{\bX}(a,e):= P_{W_{\bX}}(a).
\end{align}
Then, 
\begin{align}
f_{P_{A,E},Q_E} * g_{\bX} = f_{P_{A,E}* P_{W_{\bX}},Q_E} .
\Label{12-27-1}
\end{align}
The assumption yields that
\begin{align}
\rE_{\bX} |{\cal F}(g_{\bX})(a,e)|^2
=
\rE_{\bX} |
q^{-\frac{n}{2}} \sum_{a \in {\cal A}} 
\omega_p^{(a'\cdot a)}
P_{W_{\bX}}(a)
|^2 \le \delta^2 q^{-n}
\Label{12-27-2}
\end{align}
for $a' \neq  0 \in {\cal A}$.
Hence,
\begin{align}
& \rE_{\bX}{d_{2}}(A |E|P_{A,E} * P_{W_{\bX}} \|Q_E )
\stackrel{(a)}{=}
\rE_{\bX}\sum_{a\neq 0 \in {\cal A}, e\in {\cal E}}
|{\cal F}(f_{P_{A,E}* P_{W_{\bX}},Q_E})(a,e)|^2\nonumber \\
\stackrel{(b)}{=}
& 
\rE_{\bX}\sum_{a\neq 0 \in {\cal A}, e\in {\cal E}}
|{\cal F}(f_{P_{A,E},Q_E}* g_{\bX})(a,e)|^2 
\stackrel{(c)}{=}
\rE_{\bX} \sum_{a\neq 0 \in {\cal A}, e\in {\cal E}}
| q^{\frac{n}{2}} {\cal F}(f_{P_{A,E},Q_E}) (a,e) {\cal F}(g_{\bX}) (a,e)|^2 \nonumber\\
\stackrel{(d)}{\le}
& \delta^2 
\rE_{\bX} \sum_{a\neq 0,e} |{\cal F}(f_{P_{A,E},Q_E})(a,e) |^2 
\stackrel{(e)}{=}
\delta^2 {d_{2}}(A |E|P_{A,E} \|Q_E ) 
\le 
\delta^2 e^{-H_2(A|E|P_{A,E}\| Q_E) },
\end{align}
which shows (\ref{Lem6-1-eq1}) and (\ref{Lem6-1-eq2}).
Here, 
$(a), (b), (c), (d),$ and $(e)$ follow from 
(\ref{12-27-5}), (\ref{12-27-1}), (\ref{12-27-6}), (\ref{12-27-2}), and 
(\ref{12-27-5}), respectively.
\end{proofof}

\section{Proofs of comparisons of exponents}
\subsection{Proof of Lemma \ref{L3-29-1}}\Label{aL3-29-1}
Inequality (\ref{12-21-31}) can be shown from (\ref{8-26-8-k}).
Lemma \ref{cor1} yields that
\begin{align}
& \frac{1}{2}e_{I}(P_{A,E}|R)\nonumber \\
=& 
\max_{0 \le s \le 1} \frac{s}{2} H_{1+s}^{\downarrow}(A|E|P_{A,E} ) -\frac{s}{2}  R \nonumber \\
\le &
\max_{0 \le s \le 1} \frac{s}{2} H_{1+s}^{\uparrow}(A|E|P_{A,E} ) -\frac{s}{2}  R \nonumber \\
= &
\max_{0 \le t \le 1/2} \frac{t}{2(1-t)} (H_{\frac{1}{1-t}}^{\uparrow}(A|E|P_{A,E} ) -R) \nonumber  \\
\le &
\max_{0 \le t \le 1/2} t (H_{\frac{1}{1-t}}^{\uparrow}(A|E|P_{A,E} ) -R) 
\Label{12-22-1}\\
= & e_{d}(P_{A,E}|R),\nonumber
\end{align}
where $t= \frac{s}{1+s}$, i.e., $s=\frac{t}{1-t}$.
Inequality (\ref{12-22-1}) follows from 
the non-negativity of the RHS of (\ref{12-22-1}) and
the inequality $\frac{1}{2(1-t)} \le 1$.

\subsection{Proof of Lemma \ref{L3-19-11}}\Label{aL3-19-11}
Lemma \ref{L7-1} implies that
\begin{align*}
H_{\frac{1}{1-s}}^{\uparrow}(A|E|P_{A,E}) < H_{1+s}^{\uparrow}(A|E|P_{A,E})
\end{align*}
Choosing $t=\frac{s}{1+s}$, we have
\begin{align*}
&\max_{0\le s}\frac{s ( H_{1+s}^{\uparrow}(A|E|P_{A,E}) -R)}{1+2s}\\
=
&\max_{0\le t \le 1}\frac{ t( H_{\frac{1}{1-t}}^{\uparrow}(A|E|P_{A,E}) -R)}{1+t}\\
<
&\max_{0\le t \le 1}\frac{ t( H_{1+t}^{\uparrow}(A|E|P_{A,E}) -R)}{1+t},
\end{align*}
which implies (\ref{3-18-1}).
Similarly,
since $H_{1+t}(A|E|P_{A,E})$ is strictly monotonically increasing with respect to $t$,
\begin{align*}
&\max_{0\le s}\frac{s H_{1+s}^{\downarrow}(A|E|P_{A,E}) -sR}{1+s} \\
=&
\max_{0\le t \le 1}{t H_{\frac{1}{1-t}}(A|E|P_{A,E}) -tR} \\
< &
\max_{0\le t \le 1}{t H_{1+t}(A|E|P_{A,E}) -tR} ,
\end{align*}
which implies (\ref{3-18-1b}).

When $P_{A|E=e}$ is a uniform distribution for any element $e \in {\cal E}$,
$H_{1+t}(A|E|P_{A,E})$ and $H_{1+t}^{\uparrow}(A|E|P_{A,E})$ 
do not depend on $t$.
Hence, we obtain 
$\max_{0\le s}\frac{s ( H_{1+s}^{\uparrow}(A|E|P_{A,E}) -R)}{1+2s}
=\max_{0\le t \le 1}\frac{ t( H_{1+t}^{\uparrow}(A|E|P_{A,E}) -R)}{1+t}
=\frac{H(A|E|P_{A,E}) -R}{2}$
and
$\max_{0\le s}\frac{s H_{1+s}^{\downarrow}(A|E|P_{A,E}) -sR}{1+s} 
=\max_{0\le t \le 1}{t H_{1+t}(A|E|P_{A,E}) -tR} 
=H(A|E|P_{A,E}) -R$,
which imply
the equalities 
$e_{d}(P_{A,E}|R)=\tilde{e}_{d}(P_{A,E}|R)$
and $e_{I}(P_{A,E}|R)=\tilde{e}_{I}(P_{A,E}|R)$.

\section{Smoothing bound of min entropy}\Label{aL3-18-3}
\subsection{Proof of (\ref{3-17-4b}) of Theorem \ref{L3-18-3}}\Label{aps1}
First, 
$\Delta_{I,\min}(e^{nR},\varepsilon|P_{A,E}^n)$ is the upper bound by the smoothing of min entropy 
in the modified mutual information criterion
as is mentioned in (\ref{12-5-2-2}).
Using the relation (\ref{3-19-13}) in Theorem \ref{L3-19-10}, we obtain
\begin{align}
\liminf_{n \to \infty}\frac{-1}{n}\log
\Delta_{I,\min}(e^{nR},\varepsilon|P_{A,E}^n)
\ge 
\max_{0\le s}\frac{s H_{1+s}^{\downarrow}(A|E|P_{A,E}) -sR}{1+s}.
\Label{3-17-3b}
\end{align}
Now, we show the opposite inequality.
Applying the Cram\'{e}r Theorem \cite{Dembo},
we obtain
\begin{align}
& \lim_{n\to \infty}\frac{-1}{n}
\log 
 P_{A,E}^n
\{(a,e)\in {\cal A}^n \times {\cal E}^n|
{P_{A|E}^n(a|e)} \ge 2 e^{-nR'}
\}\nonumber\\
= &
\max_{0\le s}
sH_{1+s}^{\downarrow}(A|E|P_{A,E}) -sR'.
\Label{3-18-2b}
\end{align}
Since $sH_{1+s}^{\downarrow}(A|E|P_{A,E}) -sR'$ is monotone decreasing with respect to $R'$
and $R'-R$ is monotone increasing with respect to $R'$,
we have
\begin{align}
 \max_{R'} 
\min \{
 sH_{1+s}^{\downarrow}(A|E|P_{A,E}) -sR', R'-R \}  =
\frac{sH_{1+s}^{\downarrow}(A|E|P_{A,E})-sR}{1+s} .\Label{3-20-1}
\end{align}
because the solution of $ sH_{1+s}^{\downarrow}(A|E|P_{A,E}) -sR'= R'-R $ with respect to $R'$ is 
$\frac{sH_{1+s}^{\downarrow}(A|E|P_{A,E})+R}{1+s} $.

Using the lower bound (\ref{3-19-13b}) in Theorem \ref{L3-19-10b}
with $c=2$, 
(\ref{3-18-2b}), and (\ref{3-20-1}), 
we have
\begin{align}
&\lim_{n \to \infty}\frac{-1}{n}\log
\min_{\varepsilon>0}
(\eta( \varepsilon ,nR ) 
+ e^{nR-{H}_{\min}^{\downarrow,\varepsilon}(A|E|P_{A,E}^n )} )\nonumber\\
\le &\lim_{n \to \infty} \frac{-1}{n} \log
 \min_{R'} 
\eta( 2 P_{A,E}^n \{(a,e)\in {\cal A}^n \times {\cal E}^n|{ P_{A|E}^n(a|e)}\ge e^{-n 2 R'}  \}
 ,\log e^{n R} ) 
+ e^{n R} e^{-n R'} \nonumber \\
=&
 \max_{R'} 
\lim_{n \to \infty} \frac{-1}{n} \log
\eta( 2 P_{A,E}^n \{(a,e)\in {\cal A}^n \times {\cal E}^n|{ P_{A|E}^n(a|e)} \ge e^{-n 2 R'}  \}
 ,\log e^{n R} ) 
+ e^{n (R-R')}  \nonumber \\
=&
 \max_{R'} 
\min \{
\lim_{n \to \infty} \frac{-1}{n} \log
\eta( 2 P_{A,E}^n \{(a,e)\in {\cal A}^n \times {\cal E}^n|{ P_{A|E}^n(a|e)} \ge e^{-n 2 R'}  \}
 ,\log e^{n R} ) ,
 R'-R \}  \nonumber \\
=&
 \max_{R'} 
\min \{
\max_{0\le s} sH_{1+s}^{\downarrow}(A|E|P_{A,E}) -sR', R'-R \}  \nonumber \\
=&
 \max_{R'} 
\max_{0\le s}
\min \{
 sH_{1+s}^{\downarrow}(A|E|P_{A,E}) -sR', R'-R \}  \nonumber \\
=&
\max_{0\le s}
 \max_{R'} 
\min \{
 sH_{1+s}^{\downarrow}(A|E|P_{A,E}) -sR', R'-R \}  \nonumber \\
=&
\max_{0\le s}
\frac{sH_{1+s}^{\downarrow}(A|E|P_{A,E})-sR}{1+s} .
\end{align}
Hence, we obtain (\ref{3-17-4b}).

\subsection{Proof of (\ref{3-17-4}) of Theorem \ref{L3-18-3}}\Label{aps2}
The quantity $\Delta_{d,\min}(e^{nR},\varepsilon|P_{A,E}^n)$ 
is the upper bound by smoothing of min entropy 
in the $L_1$ distinguishability criterion
as is mentioned in (\ref{12-5-1-2}).
Using the relation (\ref{3-19-11}) in Theorem \ref{L3-19-10}, we obtain
\begin{align}
\liminf_{n \to \infty}\frac{-1}{n}\log 
\Delta_{d,\min}(e^{nR},\varepsilon|P_{A,E}^n)
\ge 
\max_{0\le s}\frac{sH_{1+s}^{\uparrow}(A|E|P_{A,E}) -sR}{1+2s}.
\Label{3-17-3}
\end{align}

We show the opposite inequality in (\ref{3-17-4}) by using the following lemma.
The proof of Lemma \ref{L3-20-1} will be shown latter.
\begin{lem}\Label{L3-20-1}
The following inequality
\begin{align}
& \lim_{n\to \infty}\frac{-1}{n}
\log 
\min_{Q_{E,n}} P_{A,E}^n
\{(a,e)\in {\cal A}^n \times {\cal E}^n |
\frac{P_{A,E}^n(a,e)}{Q_{E,n}(e)} \ge 2 e^{-nR'}
\}\nonumber \\
\le &
\max_{0\le s}
s H_{1+s}^{\uparrow}(A|E|P_{A,E}) -sR'.
\Label{3-18-2}
\end{align}
\end{lem}
Using (\ref{3-18-2}) in Lemma \ref{L3-20-1} and 
the lower bound (\ref{3-19-10b}) in 
Theorem \ref{L3-19-10b} with $c=2$, 
we obtain
\begin{align}
& \lim_{n \to \infty}\frac{-1}{n}\log
(
\min_{\epsilon_2>0} 2 \epsilon_1
+{e}^{\frac{1}{2}nR}
e^{-\frac{1}{2}{H}_{\min}^{\downarrow,\epsilon_1}(A|E|P_{A}^n )} )  \nonumber \\
\le &
\lim_{n \to \infty}\frac{-1}{n}\log
(
\min_{R'}
\min_{Q_{E,n}} P_{A,E}^n
\{(a,e)\in {\cal A}^n \times {\cal E}^n|
\frac{P_{A,E}^n(a,e)}{Q_{E,n}(e)} \ge 2 e^{-nR'}
\}
+{e}^{\frac{1}{2}n(R-R')}
) 
\nonumber \\
= &
\max_{R'}
\lim_{n \to \infty}\frac{-1}{n}\log
(
\min_{Q_{E,n}} P_{A,E}^n
\{(a,e)\in {\cal A}^n \times {\cal E}^n|
\frac{P_{A,E}^n(a,e)}{Q_{E,n}(e)} \ge 2 e^{-nR'}
\}
+{e}^{\frac{1}{2}n(R-R')}
) 
\nonumber \\
= &
\max_{R'}
\min \{\lim_{n \to \infty}\frac{-1}{n}\log
(
\min_{Q_{E,n}} P_{A,E}^n
\{(a,e)\in {\cal A}^n \times {\cal E}^n|
\frac{P_{A,E}^n(a,e)}{Q_{E,n}(e)} \ge 2 e^{-nR'}
\}),
\frac{R'-R}{2} \}
\nonumber \\
\le &
\max_{R'}
\min \{
\max_{0\le s} s H_{1+s}^{\uparrow}(A|E|P_{A,E}) -sR',
\frac{R'-R}{2} \}\nonumber \\
= &
\max_{R'}
\max_{0\le s}
\min \{
s H_{1+s}^{\uparrow}(A|E|P_{A,E}) -sR',
\frac{R'-R}{2} \} \nonumber \\
= &
\max_{0\le s}
\max_{R'}
\min \{
s H_{1+s}^{\uparrow}(A|E|P_{A,E}) -sR',
\frac{R'-R}{2} \} .
\Label{3-18-4}
\end{align}
Further, 
$sH_{1+s}^{\uparrow}(A|E|P_{A,E}) -sR'$ is monotone increasing with respect to $R'$ and
$\frac{R-R'}{2}$ is monotone decreasing with respect to $R'$.
Solving the equation
$sH_{1+s}^{\uparrow}(A|E|P_{A,E}) -sR'=\frac{R'-R}{2} $ 
with respect to $R'$,
we have $R'=\frac{2sH_{1+s}^{\uparrow}(A|E|P_{A,E})+R}{1+2s}$,
which implies that
\begin{align*}
\max_{R'}
\min \{
sH_{1+s}^{\uparrow}(A|E|P_{A,E}) -sR',
\frac{R'-R}{2} \} 
=
\frac{sH_{1+s}^{\uparrow}(A|E|P_{A,E}) -sR}{1+2s}.
\end{align*}
Thus,
\begin{align*}
& \max_{0\le s}
\max_{R'}
\min \{
sH_{1+s}^{\uparrow}(A|E|P_{A,E}) -sR',
\frac{R'-R}{2} \} \nonumber \\
=&
\max_{0\le s}
\frac{sH_{1+s}^{\uparrow}(A|E|P_{A,E}) -sR}{1+2s}.
\end{align*}
Hence, we obtain (\ref{3-17-4}).

\begin{proofof}{Lemma \ref{L3-20-1}}
We show Lemma \ref{L3-20-1} by using Lemmas \ref{L3-18-2} and \ref{L3-20-2}, which will be given latter.
For any distribution $Q_{E,n}$, we define 
the permutation invariant distribution $Q_{E,n,\inv}$ by
\begin{align*}
Q_{E,n,\inv} (e):= \sum_{g \in S_n} \frac{1}{n!} Q_{E,n}(g(e)),
\end{align*}
where $S_n$ is the $n$-th permutation group and 
$g(e)$ is the element permuted from $e \in {\cal E}^n$ by $g \in S_n$.
Then, we have
\begin{align*}
&P_{A,E}^n
\{(a,e)\in {\cal A}^n \times {\cal E}^n|
\frac{P_{A,E}^n(a,e)}{Q_{E,n}(e)} \ge 2 e^{-nR'}
\}\\
=&P_{A,E}^n
\{(a,e)\in {\cal A}^n \times {\cal E}^n|
P_{A,E}^n(a,e) \ge 2 e^{-nR'} Q_{E,n}(e)
\}\\
\ge &
\frac{1}{2} P_{A,E}^n
\{(a,e)\in {\cal A}^n \times {\cal E}^n|
P_{A,E}^n(a,e) \ge 4 e^{-nR'} Q_{E,n,\inv}(e)
\}\\
= &
\frac{1}{2} P_{A,E}^n
\{(a,e)\in {\cal A}^n \times {\cal E}^n|
\frac{P_{A,E}^n(a,e)}{Q_{E,n,\inv}(e)} \ge 4 e^{-nR'}
\},
\end{align*}
where the inequality follows from Lemma \ref{L3-18-2}.
Here, we denote the set of types of ${\cal E}$ by $T_{n,{\cal E}}$.
For any element $Q_E \in T_{n,{\cal E}}$,
we denote the uniform distribution over the subset of elements whose type is $Q_E$
by $\hat{Q}_E$.
Now, we define the distribution
\begin{align*}
Q_{E,n,\inv,\mix}(e):= \frac{1}{|T_{n,{\cal E}}|} \sum_{Q_E\in T_{n,{\cal E}}} \hat{Q}_E(e).
\end{align*}
Since $Q_{E,n,\inv}(e) \le |T_{n,{\cal E}}| Q_{E,n,\inv,\mix}(e)$,
we have
\begin{align*}
& \frac{1}{2} P_{A,E}^n
\{(a,e)\in {\cal A}^n \times {\cal E}^n|
P_{A,E}^n(a,e) \ge 4 e^{-nR'} Q_{E,n,\inv}(e)
\}\\
\ge &
 \frac{1}{2} P_{A,E}^n
\{(a,e)\in {\cal A}^n \times {\cal E}^n|
P_{A,E}^n(a,e) \ge 4 
|T_{n,{\cal E}}|
e^{-nR'} 
Q_{E,n,\inv,\mix}(e)
\}.
\end{align*}

For given sequence $(a,e) \in {\cal A} \times {\cal E}$,
we denote the type of $(a,e)$ by ${P_{A,E}'}$
and its marginal distribution over ${\cal E}$ of ${P_{A,E}'}$ by ${P_E'}$.
Then,
$P_{A,E}^n(a,e)= e^{-n (D({P_{A,E}'}\|P_{A,E})+H({P_{A,E}'}))} $
and 
$|T_{n,{\cal E}}| Q_{E,n,\inv,\mix}(e)=e^{-n H({P_E'})}$.
That is,
the condition 
$P_{A,E}^n(a,e) \ge 4 
|T_{n,{\cal E}}|
e^{-nR'} 
Q_{E,n,\inv,\mix}(e)$ is equivalent to 
the condition
$
D({P_{A,E}'}\|P_{A,E})+H({P_{A,E}'})
\le
\frac{\log 4}{n} +
H({P_E'})+R'$.
We denote the set of sequences whose types are ${P_{A,E}'}$
by $T_{P_{A,E'}}$.
Hence,
\begin{align*}
&
\frac{1}{2} P_{A,E}^n
\{(a,e)\in {\cal A}^n \times {\cal E}^n|
P_{A,E}^n(a,e) \ge 4 
|T_{n,{\cal E}}|
e^{-nR'} 
Q_{E,n,\inv,\mix}(e)
\} \\
=&
\sum_{{P_{A,E}'}\in T_{n,{\cal A}\times {\cal E}}
:
D({P_{A,E}'}\|P_{A,E})+H({P_{A,E}'})
\le
\frac{\log 4}{n} +
H({P_E'})+R'}
\frac{1}{2} P_{A,E}^n (T_{P_{A,E'}}) \\
\ge &
\max_{{P_{A,E}'}\in T_{n,{\cal A}\times {\cal E}}
:
D({P_{A,E}'}\|P_{A,E})+H({P_{A,E}'})
\le
\frac{\log 4}{n} +
H({P_E'})+R'}
\frac{1}{2} P_{A,E}^n (T_{P_{A,E}'}).
\end{align*}
Since $P_{A,E}^n (T_{P_{A,E}'})\cong e^{-n D({P_{A,E}'}\|P_{A,E})}$,
taking the limit, we have
\begin{align*}
&\lim_{n \to \infty}\frac{-1}{n}\log \frac{1}{2} P_{A,E}^n
\{(a,e)\in {\cal A}^n \times {\cal E}^n|
P_{A,E}^n(a,e) \ge 4 
|T_{n,{\cal E}}|
e^{-nR'} 
Q_{E,n,\inv,\mix}(e)
\} \\
\le &
\max_{{P_{A,E}'}}
\{D({P_{A,E}'}\|P_{A,E})
|
D({P_{A,E}'}\|P_{A,E})+H({P_{A,E}'}) \le R'+H({P_E'})
\} \\
=&
\max_{{P_{A,E}'}}
\{D({P_{A,E}'}\|P_{A,E})
|
D({P_{A,E}'}\|P_{A,E})+H(A|E|{P_{A,E}'}) \le R'
\} .
\end{align*}
Hence, combining Lemma \ref{L3-20-2},
we obtain (\ref{3-18-2}).
\end{proofof}

\begin{lem}\Label{L3-18-2}
The relation
\begin{align}
P_{A}^n \{a \in {\cal A}^n | c \ge f(a)  \}
\ge
\frac{1}{2}P_{\mix,\cA} \{a \in {\cal A}^n | c \ge \frac{1}{n!} \sum_{g\in S_n}  f(g(a))  \}
\end{align}
holds for any function $f$.
\end{lem}
\begin{proof}
Lemma \ref{L3-18-2} can be shown by applying Lemma \ref{L3-18-1} to all of distributions 
conditioned with type.
\end{proof}

\begin{lem}\Label{L3-18-1}
The relation
\begin{align}
P_{\mix,\cA} \{a| c \ge f(a)  \}
\ge
\frac{1}{2}P_{\mix,\cA} \{a| c \ge \frac{1}{|{\cal A}|}\sum_{a}f(a)  \}
\end{align}
holds for any function $f$.
\end{lem}
\begin{proof}
Markov inequality implies that
\begin{align*}
P_{\mix,\cA} \{a| c < f(a)  \}
\le
\frac{1}{c} \frac{1}{|{\cal A}|}\sum_{a}f(a) .
\end{align*}
When $c \ge \frac{2}{|{\cal A}|}\sum_{a}f(a)$,
$1- \frac{1}{c} \frac{1}{|{\cal A}|}\sum_{a}f(a) $ is 
greater than $\frac{1}{2}$.
Hence,
\begin{align*}
P_{\mix,\cA} \{a| c \ge f(a)  \}
=1-P_{\mix,\cA} \{a| c < f(a)  \}
\ge
1- \frac{1}{c} \frac{1}{|{\cal A}|}\sum_{a}f(a) 
\ge
\frac{1}{2}P_{\mix,\cA}\ \{a| c \ge \frac{2}{|{\cal A}|}\sum_{a}f(a)  \}.
\end{align*}
\end{proof}

\begin{lem}\Label{L3-20-2}
The relation
\begin{align}
&\min_{{P_{A,E}'}}
\{D({P_{A,E}'}\|P_{A,E})
|
D({P_{A,E}'}\|P_{A,E})+H(A|E|{P_{A,E}'}) \le R'
\} \nonumber \\
=&
\max_{0\le s}
sH_{1+s}^{\uparrow}(A|E|P_{A,E}) -sR'.
\Label{3-18-2bc}
\end{align}
holds.
\end{lem}

\begin{proof}
We show Lemma \ref{L3-20-2} by using Lemma \ref{L3-19-1}, which will be given latter.
We employ a generalization of the method used in \cite[Appendix D]{HKMMW}.
First, we define the distribution 
$P_{A,E;s}$ as
\begin{align*}
P_{A,E;s}(a,e):=\frac{P_{A|E}(a|e)^{1+s} P_E(e)}{(\sum_a P_{A|E}(a|e)^{1+s})^{\frac{s}{1+s}} (\sum_{e} P_E(e) (\sum_a P_{A|E}(a|e)^{1+s})^{\frac{1}{1+s}})}.
\end{align*}
That is, we have
\begin{align*}
P_{A|E;s}(a|e)&=\frac{P_{A|E}(a|e)^{1+s} }{\sum_a P_{A|E}(a|e)^{1+s}}\\
P_{E;s}(e)&= \frac{ P_E(e)(\sum_a P_{A|E}(a|e)^{1+s})^{\frac{1}{1+s}} 
}{(\sum_{e} P_E(e) (\sum_a P_{A|E}(a|e)^{1+s})^{\frac{1}{1+s}})}.
\end{align*}
Hence,
\begin{align*}
& D(P_{A,E;s}\|P_{A,E})\\
=& \sum_{a,e}P_{A,E;s}(a,e)
\Bigl(
s \log P_{A|E}(a|e)- \frac{s}{1+s}\log 
(\sum_a P_{A|E}(a|e)^{1+s} )
\Bigr)
\frac{s}{1+s}H_{1+s}^{\uparrow}( A|E|P_{A,E}), \\
& H(A|E|P_{A,E;s}) \\
=& \sum_{a,e}P_{A,E;s}(a,e)
\Bigl(
-(1+s) \log P_{A|E}(a|e) 
+\log (\sum_a P_{A|E}(a|e)^{1+s} )
\Bigr) \\
& D(P_{A,E;s}\|P_{A,E})+H(A|E|P_{A,E;s}) ,\\
=& \sum_{a,e}P_{A,E;s}(a,e)
\Bigl(
- \log P_{A|E}(a|e)+ \frac{1}{1+s}\log 
(\sum_a P_{A|E}(a|e)^{1+s} )
\Bigr)
\frac{s}{1+s} H_{1+s}^{\uparrow}( A|E|P_{A,E}) .
\end{align*}
Given $s \ge 0$, we choose an arbitrary distribution ${P_{A,E}'} $
such that
\begin{align*}
D(P_{s}^{A,E}\|P_{A,E})
=
D({P_{A,E}'}\|P_{A,E}).
\end{align*}
Since
\begin{align*}
D({P_{A,E}'}\|P_{A,E})
=&
\sum_{a,e}{P_{A,E}'}(a,e) \Bigl(\log {P_{A,E}'}(a,e) - \log P_{A,E} (a,e)\Bigr) \\
D({P_{A,E}'} \| P_{s}^{A,E})
=&
\sum_{a,e}{P_{A,E}'}(a,e) \Bigl(\log {P_{A,E}'}(a,e) - 
-(1+s)\log P_{A|E}(a|e) -\log P_E(e) \\
&+ \frac{s}{1+s} \log (\sum_a P_{A|E}(a|e)^{1+s})
-\frac{s}{1+s}H_{1+s}^{\uparrow}( A|E|P_{A,E}) 
\Bigr) ,
\end{align*}
we have
\begin{align*}
& D({P_{A,E}'} \| P_{A,E;s})
=
D({P_{A,E}'} \| P_{A,E;s})
+D(P_{A,E;s}\|P_{A,E})
-D({P_{A,E}'}\|P_{A,E}) \\
=& \sum_{a,e}P_{A,E;s}(a,e)
\Bigl(
s \log P_{A|E}(a|e)- \frac{s}{1+s}\log 
(\sum_a P_{A|E}(a|e)^{1+s} )
\Bigr)
\frac{s}{1+s}H_{1+s}^{\uparrow}(A|E|P_{A,E}) \\
&- \sum_{a,e}{P_{A,E}'}(a,e)
\Bigl(
s \log P_{A|E}(a|e)- \frac{s}{1+s}\log 
(\sum_a P_{A|E}(a|e)^{1+s} )
\Bigr)
-\frac{s}{1+s}H_{1+s}^{\uparrow}(A|E|P_{A,E}) \\
=& \sum_{a,e}
(P_{A,E;s}(a,e)-{P_{A,E}'}(a,e))
\Bigl(
s \log P_{A|E}(a|e)- \frac{s}{1+s}\log 
(\sum_a P_{A|E}(a|e)^{1+s} )
\Bigr).
\end{align*}
Hence, 
\begin{align*}
& 
H(A|E|P_{A,E;s})- H(A|E|{P_{A,E}'})+D({P_{E}'}\|P_{E;s}) \\
= & 
H(A|E|P_{A,E;s})+D(P_{A,E;s}\|P_{A,E})
- (H(A|E|{P_{A,E}'})- D({P_{A,E}'}\|P_{A,E}) )
+D({P_{E}'}\|P_{E;s}) \\
=&
\sum_{a,e}P_{A,E;s}(a,e)
\Bigl(
- \log P_{A|E}(a|e)+ \frac{1}{1+s}\log 
(\sum_a P_{A|E}(a|e)^{1+s} )
\Bigr)
+\frac{s}{1+s}H_{1+s}^{\uparrow}(A|E|P_{A,E}) \\
&-\sum_{a,e}{P_{A,E}'}(a,e)
\Bigl(
- \log P_{A|E}(a|e)+ \frac{1}{1+s}\log 
(\sum_a P_{A|E}(a|e)^{1+s} )
\Bigr)
+\frac{s}{1+s}H_{1+s}^{\uparrow}(A|E|P_{A,E}) \\
=&
\sum_{a,e}
(P_{A,E;s}(a,e)- {P_{A,E}'}(a,e))
\Bigl(
- \log P_{A|E}(a|e)+ \frac{1}{1+s}\log 
(\sum_a P_{A|E}(a|e)^{1+s} )
\Bigr) \\
=&-s D({P_{A,E}'} \| P_{A,E;s})\le 0.
\end{align*}
Since $D({P_{E}'}\|P_{E;s}) \ge 0$, 
we have $H(A|E|P_{A,E;s}) \le  H(A|E|{P_{A,E}'})$,
which implies
\begin{align*}
H(A|E|P_{A,E;s}) +D(P_{A,E;s}\|P_{A,E})
\le  H(A|E|{P_{A,E}'})+D({P_{A,E}'}\|P_{A,E}).
\end{align*}
Since the map $s \mapsto D(P_{A,E;s}\|P_{A,E})$ is continuous,
we have
\begin{align*}
&\min_{{P_{A,E}'}}
\{D({P_{A,E}'}\|P_{A,E})
|
D({P_{A,E}'}\|P_{A,E})+H(A|E|{P_{A,E}'}) \le R'
\} \\
=&
\min_{s\ge 0}
\{D({P_{A,E;s}}\|P_{A,E})
|
D({P_{A,E;s}}\|P_{A,E})+H(A|E|{P_{A,E;s}}) \le R'
\} .
\end{align*}
Now, we choose ${s_0}\ge 0$ such that
\begin{align*}
&D(P_{s_0}^{A,E}\|P_{A,E})+H(A|E|P_{s_0}^{A,E}) \\
=& \sum_{a,e}P_{s_0}^{A,E}(a,e)
\Bigl(
- \log P_{A|E}(a|e)+ \frac{1}{1+{s_0}}\log 
(\sum_a P_{A|E}(a|e)^{1+{s_0}} )
\Bigr)
+\frac{s_0}{1+s_0} H_{1+s_0}^{\uparrow}(A|E|P_{A,E}) \\
=& R',
\end{align*}
which implies that
\begin{align*}
& \sum_{a,e}P_{s_0}^{A,E}(a,e)
\Bigl(
- \log P_{A|E}(a|e)+ \frac{1}{1+{s_0}}\log 
(\sum_a P_{A|E}(a|e)^{1+{s_0}} )
\Bigr)
= R'-\frac{s_0}{1+s_0} H_{1+s_0}^{\uparrow}(A|E|P_{A,E}) .
\end{align*}
Then,
\begin{align*}
&
\min_{s\ge 0}
\{D({P_{A,E;s}}\|P_{A,E})
|
D({P_{A,E;s}}\|P_{A,E})+H(A|E|{P_{A,E;s}}) \le R'
\} \\
=& \sum_{a,e}P_{s_0}^{A,E}(a,e)
\Bigl(
{s_0} \log P_{A|E}(a|e)- \frac{{s_0}}{1+{s_0}}\log 
(\sum_a P_{A|E}(a|e)^{1+{s_0}} )
\Bigr)
+\frac{s_0}{1+s_0} H_{1+s_0}^{\uparrow}(A|E|P_{A,E})  \\
=& -{s_0} \sum_{a,e}P_{s_0}^{A,E}(a,e)
\Bigl(
-\log P_{A|E}(a|e)+ \frac{1}{1+{s_0}}\log 
(\sum_a P_{A|E}(a|e)^{1+{s_0}} )
\Bigr)
+\frac{s_0}{1+s_0} H_{1+s_0}^{\uparrow}( A|E|P_{A,E})  \\
=&-{s_0}( R'+ \phi (\frac{{s_0}}{1+{s_0}}| A|E|P_{A,E}) )
+\frac{s_0}{1+s_0} H_{1+s_0}^{\uparrow}(A|E|P_{A,E})  \\
=&-{s_0} R'+s_0 H_{1+s_0}^{\uparrow}(A|E|P_{A,E}) \\
=& \max_{s \ge 0}-s R'+s H_{1+s}^{\uparrow}(A|E|P_{A,E}) ,
\end{align*}
where the reason of the equation is the following.
Due to Lemma \ref{L3-19-1}, the function 
$s \mapsto -s H_{1+s}^{\uparrow}(A|E|P_{A,E}) $ 
is convex,
and $-R'= -\frac{d}{ds}s H_{1+s}^{\uparrow}(A|E|P_{A,E})$.
Then, we obtain (\ref{3-18-2bc}).
\end{proof}

\bibliographystyle{IEEE}

\end{document}